\documentclass[10pt,journal,twocolumn]{IEEEtran}
\usepackage[T1]{fontenc}
\usepackage{amsmath,amsfonts}
\usepackage{algorithmic}
\usepackage{array}
\usepackage[caption=false,font=normalsize,labelfont=sf,textfont=sf]{subfig}
\usepackage{textcomp}
\usepackage{stfloats}
\usepackage{hyperref}       
\hypersetup{
  colorlinks   = true, 
  urlcolor     = blue, 
  linkcolor    = blue, 
  citecolor   = blue 
}
\usepackage{url}
\usepackage{verbatim}
\usepackage{graphicx}
\def\BibTeX{{\rm B\kern-.05em{\sc i\kern-.025em b}\kern-.08em
    T\kern-.1667em\lower.7ex\hbox{E}\kern-.125emX}}
\usepackage{balance}
\usepackage{bbm}
\usepackage{physics}
\newtheorem{theorem}{Theorem}[section]
\newtheorem{corollary}[theorem]{Corollary}
\newtheorem{lemma}[theorem]{Lemma}
\newtheorem{remark}[theorem]{Remark}
\newtheorem{definition}[theorem]{Definition}
\DeclareMathOperator{\Prob}{Prob}

\DeclareMathOperator{\id}{id}
\DeclareMathOperator{\Loc}{Loc}
\DeclareMathOperator{\Spec}{Spec}

\DeclareMathOperator{\AVP}{AVP}

\DeclareMathOperator{\Dec}{\Gamma^{\mathcal{D}}}
\DeclareMathOperator{\DecI}{Dec}
\DeclareMathOperator{\Enc}{\Gamma^{\mathcal{E}}}
\DeclareMathOperator{\EncI}{Enc}

\DeclareMathOperator{\EC}{EC}
\DeclareMathOperator{\Dist}{Dist}

\DeclareMathOperator{\Dil}{Dil}
\DeclareRobustCommand
  \Compactcdots{\mathinner{\cdotp\mkern-1mu\cdotp\mkern-1mu\cdotp}}

\begin{document}

\title{Fault-tolerant Coding for Entanglement-Assisted Communication}

\author{Paula Belzig, Matthias Christandl, Alexander Müller-Hermes
\thanks{PB (paulabelzig@gmail.com) and MC (christandl@math.ku.dk) are with the Department of Mathematical Sciences,
University of Copenhagen, Universitetsparken 5, 2100 Copenhagen,
Denmark.
}
\thanks{AMH (muellerh@math.uio.no) is with the Department of
Mathematics, University of Oslo, P.O. box 1053, Blindern, 0316 Oslo,
Norway.}
\thanks{PB and MC acknowledge financial support from the European Research Council (ERC Grant Agreement no.\  81876), VILLUM FONDEN via the QMATH Centre of Excellence (grant no.\ 10059) and the Novo Nordisk Foundation (grant NNF20OC0059939 ‘Quantum for Life’). PB acknowledges funding from the European Union’s Horizon 2020 research and innovation programme under the Marie Skłodowska-Curie TALENT Doctoral fellowship (grant no.\ 801199). AMH acknowledges funding from the European Union’s Horizon 2020 
research and innovation programme under the Marie
Skłodowska-Curie Action TIPTOP (grant no.\ 843414) and by The Research 
Council of Norway (project\ 324944).

A part of this work has been presented at ISIT 2023 \cite{BCMH23SHORT}.}
}



\maketitle

\begin{abstract}
Channel capacities quantify the optimal rates of sending information reliably over noisy channels. Usually, the study of capacities assumes that the circuits which the sender and receiver use for encoding and decoding consist of perfectly noiseless gates. In the case of communication over quantum channels, however, this assumption is widely believed to be unrealistic, even in the long-term, due to the fragility of quantum information, which is affected by the process of decoherence. Christandl and Müller-Hermes have therefore initiated the study of fault-tolerant channel coding for quantum channels, i.e. coding schemes where encoder and decoder circuits are affected by noise, and have used techniques from fault-tolerant quantum computing to establish coding theorems for sending classical and quantum information in this scenario. Here, we extend these methods to the case of entanglement-assisted communication, in particular proving that the fault-tolerant capacity approaches the usual capacity when the gate error approaches zero. A main tool, which might be of independent interest, is the introduction of fault-tolerant entanglement distillation. We furthermore focus on the modularity of the techniques used, so that they can be easily adopted in other fault-tolerant communication scenarios.
\end{abstract}


\section{Introduction}

\label{sec-intro}

\IEEEPARstart{T}{he} successful transfer of information via a communication infrastructure is of crucial importance for our modern, highly-connected world. This process of information transfer, e.g., by wire, cable or broadcast, can be modelled by a communication channel $T$ which captures the noise affecting individual symbols. Instead of sending symbols individually, the sender and receiver typically agree to send messages using codewords made up from many symbols. With a well-suited code, the probability of receiving a wrong message can be made arbitrarily small. How well a given channel $T$ is able to transmit information can be quantified by the asymptotic rate of how many message bits can be transmitted per channel use with vanishing error using the best possible encoding and decoding procedure. This asymptotic rate is a characteristic of the channel, called its \textit{capacity} $C(T)$.

In \cite{Shannon48}, Shannon introduced this model for communication and derived a formula for $C(T)$ in terms of the mutual information between the input and output of the channel:
\[C(T)=\sup_{p_X} I(X:Y).\]
Here, $I(X:Y)=H(X)+H(Y)-H(XY)$ denotes the mutual information between the random variable $X$ and the output $Y=T(X)$, where $H(X)=-\sum_x p_X(x)\log(p_X(x) )$ is the Shannon entropy of the discrete random variable $X$ with a set of possible values $x$ that is distributed according to a probability distribution $p_X$.

Various generalizations of this communication scenario to a quantum channel $T:\mathcal{M}_{d_A} \rightarrow \mathcal{M}_{d_B}$, where $\mathcal{M}_d$ denotes the matrix algebra of complex $d\times d$-matrices, lead to different notions of capacity. Two important examples are the classical capacity of a quantum channel \cite{Holevo96,SW97}, which quantifies how well a quantum channel can transmit classical information encoded in quantum states, and the quantum capacity \cite{Lloyd97,Shor02,Devetak03}, where quantum information itself is to be transmitted through the channel. Both of these notions of capacity have entropic formulas. However, they are not known to admit a characterization which is independent of the number of channel copies, a so-called single-letter characterization, which would simplify their calculation.

The \textit{entanglement-assisted capacity}, where the encoding and decoding machines have access to arbitrary amounts of entanglement, does not only admit such a single-letter characterization, but it can in fact be regarded as the only direct formal analogue of Shannon's original formula, since the classical mutual information is simply replaced by its quantum counterpart \cite{BSST02}:
\[C^{ea}(T)= \sup_{\substack{\varphi \in \mathcal{M}_{d_A}\otimes\mathcal{M}_{d_{A'}}  \\  \varphi \text{ a pure quantum state} }} I(A':B)_{(T\otimes \id_{{A'}})(\varphi)} .\]
Here, $I(A:B)_{\rho}=H(A)_{\rho}+H(B)_{\rho}-H(AB)_{\rho}$ denotes the quantum mutual information with the von Neumann entropy $H(A)_{\rho}=-\Tr\left[\rho\log(\rho)\right]$ for a quantum state $\rho\in\mathcal{M}_{d_A}$.

In order to communicate with a given channel $T$, the encoding and decoding procedures need to be decomposed into quantum circuits as a sequence of quantum gates. The next step in a real-world scenario would be to implement these circuits on a quantum device so that we can realize an actual quantum communication system. However, this scenario generally does not consider one of the major obstacles of quantum computation: the high susceptibility of quantum circuits to noise and faults. In classical computers, the error rates of individual logical gates are known to be effectively zero in standard settings and at the time-scales relevant for communication \cite{Nicolaidis11}. The assumption of noiseless gates implementing the encoder and decoder circuit is therefore realistic in many scenarios. Real-life quantum gates, however, are affected by non-negligible amounts of noise. This is certainly a problem in near-term quantum devices, and it is generally assumed that it will continue to be a problem in the longer term \cite{Preskill18}.

Considering the encoder and decoder circuits as specific quantum circuits affected by noise therefore leads to potentially more realistic measures of how well information can be transferred via a quantum channel: \textit{fault-tolerant capacities}, which quantify the optimal asymptotic rates of transmitting information per channel use in the presence of noise on the individual gates. To construct suitable encoders and decoders for this scenario, we build on Christandl and Müller-Hermes' work \cite{CMH20}, which has introduced and analyzed fault-tolerant versions of the classical and quantum capacity, combining techniques from fault-tolerant quantum computing \cite{AB99,KLZ98,Kitaev03,AGP05} and quantum communication theory \cite{Wilde13}.


More precisely, we extend their work to entanglement-assisted communication. In particular, we show that entanglement-assisted communication is still possible under the assumption of noisy quantum devices, with achievable rates given by \[C^{ea}_{\mathcal{F}(p)}(T) \geq C^{ea}(T)-f(p)\]
where $C^{ea}_{\mathcal{F}(p)}(T)$ denotes the fault-tolerant entanglement-assisted capacity for gate error probability $p$ below a threshold, and with $\lim_{p\rightarrow 0} f(p) \rightarrow 0$.


In other words, the achievable rates for entanglement-assisted communication with noise-affected gates can be bounded from below in terms of the quantum mutual information reduced by a continuous function in the single gate error $p$. The usual faultless entanglement-assisted capacity is recovered for small probabilities of local gate error, which confirms and substantiates the practical relevance of quantum Shannon theory. 
This is not only relevant for communication between spatially separated quantum computers, but also for communication between distant parts of a single quantum computing chip, where the communication line may be subject to higher levels of noise than the local gates. In particular, the noise level for the communication line does not have to be below the threshold of the gate error.

It is important to note that many of the existing techniques from quantum fault-tolerance cannot directly be applied to the problem of communication, or will only allow for weaker results. Naive strategies with one (large) fault-tolerant implementation, where the communication channel is considered as part of the circuit noise, will only give rates approaching zero due to their high overhead implementations, and they will only work for channels which are very close to the identity, i.e. with noise below the threshold. In this work and for the results above, we are not only interested in transmitting with vanishing error, but also at communication rates that are as high as possible and for comparatively noisy channels.

The manuscript is structured around the building blocks needed to achieve this result. In Section~\ref{sec-ft}, we briefly review concepts from fault-tolerance of quantum circuits used for communication. In Section~\ref{sec-ea-cap}, we outline how the fault-tolerant communication setup can be reduced to an information-theoretic problem which generalizes the usual, faultless entanglement-assisted capacity. In Section~\ref{sec-avp}, we prove a coding theorem for this information-theoretic problem. One important facet of communication with entanglement-assistance in our scenario comes in the form of noise affecting the entangled resource states, for which we introduce a scheme of fault-tolerant entanglement distillation in Section~\ref{sec-ft-ent-dist}. Finally, these techniques will be combined to obtain a threshold-type coding theorem for fault-tolerant entanglement-assisted capacity in Section~\ref{sec-coding-thm}.

\section{Fault-tolerant encoder and decoder circuits for communication}
\label{sec-ft}

\noindent Here, we review some aspects of common techniques for fault-tolerance, but for a detailed overview of the relevant concepts, we refer to \cite{AGP05} and \cite{CMH20}.

Note that our notation for mathematical objects from quantum theory is the same as in \cite[Section~II.A]{CMH20}. We define quantum channels as completely positive and trace preserving maps $T: \mathcal{M}_{d_A}\rightarrow \mathcal{M}_{d_B}$ where $\mathcal{M}_d$ denotes the matrix algebra of complex $d\times d$-matrices.
Probability distributions of $d$ elements are vectors in $\mathbbm{C}^d$ where each entry is positive and the sum of all entries equals $1$. Channels with classical input are defined as linear maps from $\mathbbm{C}^d$ that yield unit-trace positive semi-definite Hermitian matrices, and channels with classical output map unit-trace positive semi-definite Hermitian matrices to elements of $\mathbbm{C}^d$.

\subsection{Fault-tolerance for quantum circuits}
\label{sec-ft-circuits}

\noindent Quantum circuits are the dense subset of quantum channels which can be written as a composition of the following elementary gate operations: identity gate, Pauli gates, Hadamard gate H, T-gate, CNOT gate, discarding of a qubit (i.e. performing a trace of a subsystem), and measurements and preparations in the computational basis. It should be emphasized that the linear map realized by a quantum circuit might be written in different ways as a composition of elementary gates. As in \cite{CMH20}, we will assume that each circuit is specified by a particular circuit diagram detailing which elementary gates are to be executed at which time and place in the quantum circuit. The set of elementary operations in the circuit diagram of a circuit $\Gamma$ is the circuit's \emph{set of locations} $\Loc(\Gamma)$, and the number of elementary operations in the decomposition is denoted by $| \Loc(\Gamma)|$.

Given such a circuit diagram, we can model the noise affecting the resulting quantum circuit. For simplicity, we will always consider the \textit{i.i.d. Pauli noise model}, where one of the Pauli channels (with single Kraus operator $\sigma_x,\sigma_y$ or $\sigma_z$) is applied with probability $\frac{p}{3}$ in between the gates. Specifically, we use the following convention (as in \cite{AGP05}): For operations acting on a single qubit, a single Pauli channel is applied before (in case of a measurement or trace gate) or after (in case of single qubit gates and preparation gates) the gate itself; in case of the CNOT gate, a tensor product of two Pauli channels is applied after the CNOT gate (see also \cite[Definition~II.1]{CMH20}). This is a common and well-motivated noise model \cite{NC00}, further supplemented by comparison between experiment and classical simulations \cite{Google21}. Here, we choose to limit our work to the Pauli i.i.d. noise model in order to simplify the presentation, but we see no obstacles in extending our results to stochastic i.i.d. noise models. For more general noise models, different techniques may be required.



The pattern in which noise occurs (i.e. the location where a Pauli channel is inserted, and which Pauli channel) is specified by a \textit{fault pattern} $F$ and a quantum circuit $\Gamma$ which is affected by noise according to a fault pattern $F$ is denoted by $[\Gamma]_F$. We can expand the linear map represented by a fault-affected quantum circuit in our model as a sum over Pauli-fault patterns: $[\Gamma]_{\mathcal{F}(p)} =\sum_{F\in\mathcal{F}(p)} P(F)[\Gamma ]_F $, where a Pauli fault pattern $F$ occurs with a probability $P(F)$ according to the probability distribution specified by $\mathcal{F}(p)$. In case of the i.i.d. Pauli noise model, each fault pattern $F$ occurs with a classical probability $P(F)=(1-p)^{l_{\id}}(p/3)^{l_x+l_y+l_z}$ where $l_{k}$ is the number of locations where a fault appears with each index corresponding to a type of Pauli-fault $k=\{ \id, x, y,z\}$.

To protect against noise, a quantum circuit can be implemented in a stabilizer error correcting code, where single, potentially fault-affected qubits (\textit{logical qubits}) can be encoded in a quantum state of $K$ physical qubits for each logical qubit. Let $\Pi_K$ be the group of all $K$-fold tensor products of the Pauli matrices $\sigma_x,\sigma_y,\sigma_z$ and $\mathbbm{1}$. A stabilizer code is obtained by selecting a commuting subgroup of $\Pi_n$ that does not contain $-\mathbbm{1}^{\otimes K}$ (called the \textit{stabilizer group}), and has an associated simultaneous $+1$-eigenspace $\mathcal{C}\in(\mathbbm{C}^{2})^{\otimes K}$, which is called the \textit{code space}. Here, we will assume this subspace to have dimension $2$, i.e., we encode a single qubit, where the stabilizer subgroup is generated by $K-1$ elements. We will denote these elements by $g_1,\ldots ,g_{K-1}$. 

Any product Pauli operator $E:(\mathbbm{C}^{2})^{\otimes K}\rightarrow (\mathbbm{C}^{2})^{\otimes K}$ either commutes or anti-commutes with elements of this stabilizer group and can therefore be associated to a vector
\[
s = (s_1,\ldots , s_{K-1})\in\mathbbm{F}^{K-1}_2 ,
\]
where $s_i=0$ if $E$ commutes with $g_i$ or $s_i=1$ if it anti-commutes with $g_i$. We will call the vector $s$ the syndrome associated to the Pauli operator $E$ and it is essentially the quantity that is measured when performing error correction with the stabilizer code.

A general quantum state can be decomposed in terms of eigenspaces associated to the syndromes, as
\[
(\mathbbm{C}^{2})^{\otimes K} = \bigoplus_{s\in\mathbbm{F}^{K-1}_2} W_s ,
\]
where $W_s$ is the common eigenspace of the operators $g_1,\ldots ,g_{K-1}$ where we have an eigenvalue $(-1)^{s_i}$ for $g_i$ for each $i$. Each $W_s$ is $2$-dimensional and can be associated to some Pauli operator $E_s$ such that 
\[
W_s = \text{span}\left\lbrace E_s\ket{\overline{0}},E_s\ket{\overline{1}}\right\rbrace,
\]
where $\left\lbrace\ket{\overline{0}},\ket{\overline{1}}\right\rbrace$ are the logical $\ket{0}$ and $\ket{1}$. 
Then, we can choose a decoder by defining a unitary transformation $D:(\mathbbm{C}^{2})^{\otimes K}\rightarrow \mathbbm{C}^2\otimes (\mathbbm{C}^{2})^{\otimes (K-1)}$ such that 
\[
D\left( E_s\ket{\overline{b}}\right) = \ket{b}\otimes \ket{s},
\]
for any $b\in \left\lbrace 0,1\right\rbrace$ and any $s\in\mathbbm{F}^{K-1}_2$. In principle, several choices of $E_s$ can be associated to a syndrome $s$, and the choice of the basis change $D$ singles out specific Pauli-errors which constitute the set of correctable errors of our code.

With this unitary, we define the \emph{ideal decoder} $\text{Dec}^*:\mathcal{M}^{\otimes K}_2\rightarrow \mathcal{M}_2\otimes \mathcal{M}^{\otimes (K-1)}_2$ given by $\text{Dec}^*(X) = DXD^\dagger$. We also define its inverse, the \emph{ideal encoder} $\text{Enc}^*:\mathcal{M}_2\otimes \mathcal{M}^{\otimes (K-1)}_2\rightarrow \mathcal{M}^{\otimes K}_2$. Finally, we define the ideal error correcting channel as the following object:
\[\EC^*=\EncI^*\circ (\id_2 \otimes \ketbra{0}{0}^{\otimes K-1} \Tr )\circ \DecI^*\]
where the second system corresponds to the syndrome state on the syndrome space $\mathcal{M}_2^{\otimes K-1}$. %
The syndrome state $\ketbra{0}{0}^{\otimes K-1}$ corresponds to the zero syndrome where $E_0 = \mathbbm{1}^{\otimes K}_2$. See also \cite[Section~II.C]{CMH20} for a detailed discussion of these ideal quantum channels.

It should be emphasized that the ideal encoder and decoder are not physical operations. They only appear as mathematical tools when analyzing noisy quantum circuits encoded in the stabilizer code. The output space of the ideal decoder is written as $\mathcal{M}_2\otimes \mathcal{M}^{\otimes (K-1)}_2$ to emphasize that we think differently about these two tensor factors, and we sometimes refer to the first one as the \emph{logical space}, and to the second one as the \emph{syndrome space}.

Throughout this work, we will frequently use notation where an operation marked with a star should be considered an ideal operation that is useful for circuit analysis, and not a fault-location. In particular, we will sometimes write $\id_2^*$ to denote an identity map between qubits, which should not be taken to be a fault-affected storage.

Gates on logical qubits are implemented as so-called gadgets on physical qubits, using the operations $\EncI^*$ and $\DecI^*$ to map between the spaces (see also \cite[Definition~II.3]{CMH20}). For a circuit $\Gamma:\mathcal{M}_2^{\otimes n}\rightarrow \mathcal{M}_2^{\otimes m}$, its implementation in a code $\mathcal{C}$, $\Gamma_{\mathcal{C}}:\mathcal{M}_2^{\otimes nK}\rightarrow\mathcal{M}_2^{\otimes mK}$, is obtained by replacing each gate by its corresponding gadget and inserting error correction gadgets in between the gadgets.
If the physical qubits of this implementation are subject to the noise model $\mathcal{F}(p)$, then we denote this fault-affected implementation of the circuit $\Gamma$ by $[\Gamma_{\mathcal{C}}]_{\mathcal{F}(p)}:\mathcal{M}^{\otimes nK}\rightarrow\mathcal{M}^{\otimes mK}$. 

In this work, like in \cite{CMH20}, we will consider implementations in the concatenated 7-qubit Steane code. The 7-qubit Steane code introduced in \cite{Steane96} is an error correcting code that can correct all single-qubit errors, and that can be concatenated to improve protection against errors \cite{KL96}. For the concatenated 7-qubit Steane code, as shown in \cite{AGP05}, an implementation where the error correction's gadget is performed between each operation minimizes the accumulation of errors. Under this implementation, the concatenated 7-qubit Steane code fulfills a threshold theorem for computation. More precisely, it has been shown that the difference between a quantum circuit $\Gamma$ with classical input and output and its fault-affected implementation $[\Gamma_{\mathcal{C}}]_{\mathcal{F}(p)}$ in the concatenated 7-qubit Steane code is bounded by $p_0 \Big(\frac{p}{p_0}\Big)^{2^l} |\Loc(\Gamma)|$ for any $p$ below a threshold $p_0$ and concatenation level $l$ \cite{AGP05, CMH20}. By choosing the level $l$ large enough, this error can be made arbitrarily small.

Fault-tolerance can in principle be achieved by other quantum error correcting codes \cite{AB99,KLZ98,Kitaev03,AGP05}. One could also consider using two different quantum error correcting codes for the encoder and decoder circuit in our setup. For simplicity, we restrict ourselves to using the concatenated 7-qubit Steane code \cite{Steane96,KL96} with the same level of concatenation for both circuits, but our definitions can straightforwardly be extended to the more general case.

\subsection{Fault-tolerance for communication}
\label{sec-ft-interfaces}

\noindent By performing error correction, a quantum circuit with classical input and output that is affected by faults at a low rate can thus be implemented in a way such that it behaves like an ideal circuit (i.e. a circuit without faults) by threshold-type theorems. These code implementations cannot, however, be directly used in the encoding and decoding circuits for communication, as they require classical input and output, whereas the encoder's output in our communication setup, for instance, serves as input into the noisy quantum channel. The fault-tolerant implementation of an encoder and decoder in a communication setting therefore leads to the message being encoded in the corresponding code space. In the case of the concatenated 7-qubit Steane code, the number of physical qubits increases by a factor of $7$ for each level of concatenation. To obtain our results for communication rates, we therefore perform an additional circuit mapping information in the code space to the physical system where the quantum channel acts. This circuit will be referred to as \textit{decoding interface} $\DecI$. Similarly, another circuit can be performed to transfer the channel's output into the code space where it can be processed by the fault-tolerantly implemented decoder. This circuit is called \textit{encoding interface} $\EncI$. These circuits, introduced in Definition~\ref{def-interfaces}, are also affected by faults.

\begin{definition}[{Interfaces, \cite[Definition III.1]{CMH20}}] \label{def-interfaces}
 Let $\mathcal{C}\in (\mathbbm{C}^2)^{\otimes K}$ be a stabilizer code with $\dim(\mathcal{C})=2$, and let $\ketbra{0}{0}\in \mathcal{M}_2^{\otimes K-1}$ denote the state corresponding to the zero-syndrome. Let $\EncI^*:\mathcal{M}_2^{\otimes K}\rightarrow\mathcal{M}_2^{\otimes K}$ and $\DecI^*:\mathcal{M}_2^{\otimes K}\rightarrow\mathcal{M}_2^{\otimes K}$ be the ideal encoding and decoding operations. Then, we have:
\begin{enumerate}
    \item An encoding interface $\EncI:\mathcal{M}_2\rightarrow\mathcal{M}_2^{\otimes K}$ for a code $\mathcal{C}$ is a quantum circuit with an error correction as a final step, and fulfilling \[\DecI^*\circ \EncI =\id_2 \otimes \ketbra{0}{0} \]
    \item A decoding interface $\DecI:\mathcal{M}_2^{\otimes K}\rightarrow\mathcal{M}_2$ is a quantum circuit fulfilling \[\DecI \circ \EncI^* (\cdot \otimes \ketbra{0}{0}) =\id_2 (\cdot) \]
    \end{enumerate}
\end{definition}

In contrast to $\DecI^*$ and $\EncI^*$, which are objects used for the mathematical analysis of the circuits and not implemented in practice, the interfaces $\EncI$ and $\DecI$ are quantum circuits consisting of gates that can be affected by faults. Since this can lead to faulty inputs to a quantum channel, we will need interfaces that are tolerant against such faults. Unfortunately, it is impossible to make the overall failure probability of interfaces arbitrarily small, since they will always have a first (or last) gate that is executed on the physical level and not protected by an error correcting code, resulting in a failure with a probability of at least gate error $p$. Fortunately, it is possible to construct qubit interfaces for concatenated codes which fail with a probability of at most $2cp$ for some constant $c$, which are good enough for our purposes~\cite{MGHHLPP14,CMH20}.

\begin{theorem}[{Correctness of interfaces for the concatenated 7-qubit Steane code, \cite[Theorem III.3]{CMH20}}] \label{thm-correct-interfaces}
For each $l\in \mathbbm{N}$, let ${\mathcal{C}_l}$ denote the $l$-th level of the concatenated 7-qubit Steane code with threshold $p_0$. Then,
there exist interface circuits $\EncI_l:\mathcal{M}_2 \rightarrow \mathcal{M}_2^{\otimes 7^l}$ and $\DecI_l:\mathcal{M}_2^{\otimes 7^l}\rightarrow\mathcal{M}_2$  for the $l$-th level of this code such that for any $0\leq p \leq \frac{p_0}{2}$, we have
\begin{enumerate}
    \item \[\Prob( \big[\EncI_l \big]_F \text{ is not correct}) \leq 2cp,\] where $\EncI_l$ is correct under a Pauli fault pattern $F$ if there exists a quantum state $\sigma_S(F)$ on the syndrome space such that $\DecI^*\circ \big[\EncI \big]_F = \id_2 \otimes \sigma_S(F)$. The probability is taken over the distribution of $F$ according to the fault model $\mathcal{F}(p)$.
    \item \[\Prob( \big[\DecI_l \big]_F \text{ is not correct}) \leq 2cp,\] where $\DecI$ is correct under a Pauli fault pattern $F$ if $  \big[\DecI_l\circ \EC_l \big]_F = (\id_2 \otimes \Tr_S)\circ \DecI_l^*\circ  \big[\EC_l \big]_F $ where $\Tr_S$ traces out the syndrome space.
\end{enumerate}
Here, $c=p_0 \max \{|\Loc(\EncI_1)|,|\Loc(\DecI_1\circ \EC)|\}$ is a constant that does not depend on $l$ or $p$.
\end{theorem}

In combination with the threshold theorem from \cite{AGP05}, this can be used to prove extensions of Lemma III.8 from \cite{CMH20} for the combination of circuit and interface with additional quantum input (cf. Figure~\ref{fig-eff-channel}). 
Here, $\id_{cl}:\mathbbm{C}^2\rightarrow \mathbbm{C}^2$ denotes the identity map on a classical bit, and $\| T-S \|_{1 \rightarrow 1} := \sup \{ \| (T-S)(\rho) \|_{\Tr}   | \rho\in \mathcal{M}_{d_A}\text{a quantum state}\}$ denotes the 1-to-1 distance of two quantum channels $T:\mathcal{M}_{d_A}\rightarrow \mathcal{M}_{d_B} $ and $S:\mathcal{M}_{d_A}\rightarrow \mathcal{M}_{d_B} $, where $ \| \rho - \sigma \|_{\Tr} $ denotes the trace distance induced by the trace norm $\| \rho \|_{\Tr} := \frac{1}{2} \| \rho\|_{1} =\frac{1}{2}  \Tr(\sqrt{ \rho ^{\dagger} \rho}) $. 

These lemmas, and the subsequent effective channel model in Theorem \ref{thm-eff-channel} are key ingredients to the analysis of fault-tolerant communication because they allow us to connect the faulty encoder and decoder circuits of the communication scheme with effective, fault-less versions of the encoder and decoder circuits in an effective communication problem.

\begin{lemma}[Effective encoding interface]
\label{thm-effective-encoder}
Let $m,n,k\in \mathbbm{N}$ and let $\Gamma:\mathcal{M}_2^{\otimes n+k} \rightarrow \mathbbm{C}^{2^m}$ be a quantum circuit with quantum input and classical output.
For each $l\in \mathbbm{N}$, let ${\mathcal{C}_l}$ denote the $l$-th level of the concatenated 7-qubit Steane code with threshold $p_0$. Moreover, let $\EncI_l: \mathcal{M}_2 \rightarrow  \mathcal{M}_2^{\otimes 7^l}$ 
be the encoding interface circuit for the $l$-th level of the concatenated 7-qubit Steane code with threshold $p_0$.

Then, for any $0\leq p \leq \frac{p_0}{2}$ and any $l\in\mathbbm{N}$, there exists a quantum channel $N_l:\mathcal{M}_2\rightarrow \mathcal{M}_2$, which only depends on $l$ and the interface circuit $\EncI_l$, such that:
\begin{IEEEeqnarray*}{Cl}
 &\|  \big[\Gamma_{\mathcal{C}_l} \circ (\EncI_l^{\otimes n} \otimes \EC_l^{\otimes k}) \big]_{\mathcal{F}(p)}  \\&\hspace{0.25cm}-(\Gamma\otimes \Tr_S) \circ \left(N_{enc,p,l}^{\otimes n} \otimes (\DecI_l^*\circ [\EC_l]_{\mathcal{F}(p)} )^{\otimes k})\right) \|_{1 \rightarrow 1}  \\ &\leq 2 p_0 \Big(\frac{p}{p_0}\Big)^{2^l} |\Loc(\Gamma)| 
\end{IEEEeqnarray*}
with
\[N_{enc,p,l} = (1-2cp) \id_2 +2cpN_l\]
where $c=p_0 \max \{|\Loc(\EncI_1)|,|\Loc(\DecI_1\circ \EC)|\}$.
\end{lemma}

\begin{lemma}[Effective decoding interface]
\label{thm-eff-decoder}
Let $m,n,k\in \mathbbm{N}$ and $\Gamma: \mathbbm{C}^{2^m} \otimes \mathcal{M}_2^{\otimes k}\rightarrow \mathcal{M}_2^{\otimes n} $ be a quantum circuit with quantum and classical input and quantum output.
For each $l\in \mathbbm{N}$, let ${\mathcal{C}_l}$ denote the $l$-th level of the concatenated 7-qubit Steane code with threshold $p_0$.  Moreover, let 
$\DecI_l: \mathcal{M}_2^{\otimes 7^l} \rightarrow \mathcal{M}_2$ be the decoding interface circuit for the $l$-th level of the concatenated 7-qubit Steane code with threshold $p_0$.

Then, for any $0\leq p \leq \frac{p_0}{2}$ and any $l\in\mathbbm{N}$, there exists a quantum channel $N_l:\mathcal{M}_2 \otimes \mathcal{M}_2^{\otimes(7^l-1)} \rightarrow \mathcal{M}_2$, which only depends on $l$ and the interface circuit $\DecI_l$, such that:
    \begin{IEEEeqnarray*}{Cl}
     \IEEEeqnarraymulticol{2}{l}{  \| \big[ \DecI_l^{\otimes n}  \circ \Gamma_{\mathcal{C}_l} \circ (\id_{cl}^{\otimes m} \otimes \EC_l^{\otimes k}) \big]_{\mathcal{F}(p)} } \\& -N_{dec,p,l}^{\otimes n} \circ \left(\Gamma\otimes S_{S} \right)\circ (\id_{cl}^{\otimes m} \otimes (\DecI_l^* \circ [\EC_l]_{\mathcal{F}(p)})^{\otimes k}) \|_{1 \rightarrow 1} \\ \leq& 2 p_0 \Big(\frac{p}{p_0}\Big)^{2^l} |\Loc(\Gamma)| + 2np_0|\Loc(\EncI_1)| \Big(\frac{p}{p_0}\Big)^{2^l-1}  ,
    \end{IEEEeqnarray*}
where $S_S: \mathcal{M}_2^{\otimes k(7^l-1) }\rightarrow  \mathcal{M}_2^{\otimes n(7^l-1) }$ is some quantum channel on the syndrome space, and with
\[N_{dec,p,l} = (1-2cp) \id_2 \otimes \Tr_{S} +2cp N_l\]
where $c=p_0 \max \{|\Loc(\EncI_1)|,|\Loc(\DecI_1\circ \EC)|\}$.
\end{lemma}


\begin{IEEEproof}[Proof sketch for Lemma~\ref{thm-effective-encoder} and \ref{thm-eff-decoder}]
  When the circuit $\Gamma_{\mathcal{C}_l}$ receives quantum input in the code space, the transformation rules of the circuit elements that ensure fault-tolerance refer to the circuit element as well as the preceding error correction. Therefore, an $\EC$ gadget is included in the statement for the parts of the circuit that receive quantum input. Then, like the proof of  \cite[Lemma~IIII.8]{CMH20}, this adapted statement follows from a repeated application of the transformation rules from \cite[Lemma~4]{AGP05} and \cite[Lemma~II.6]{CMH20}. 
\end{IEEEproof}

\begin{figure*}[!t]
  \centering
       \includegraphics[width=\textwidth]{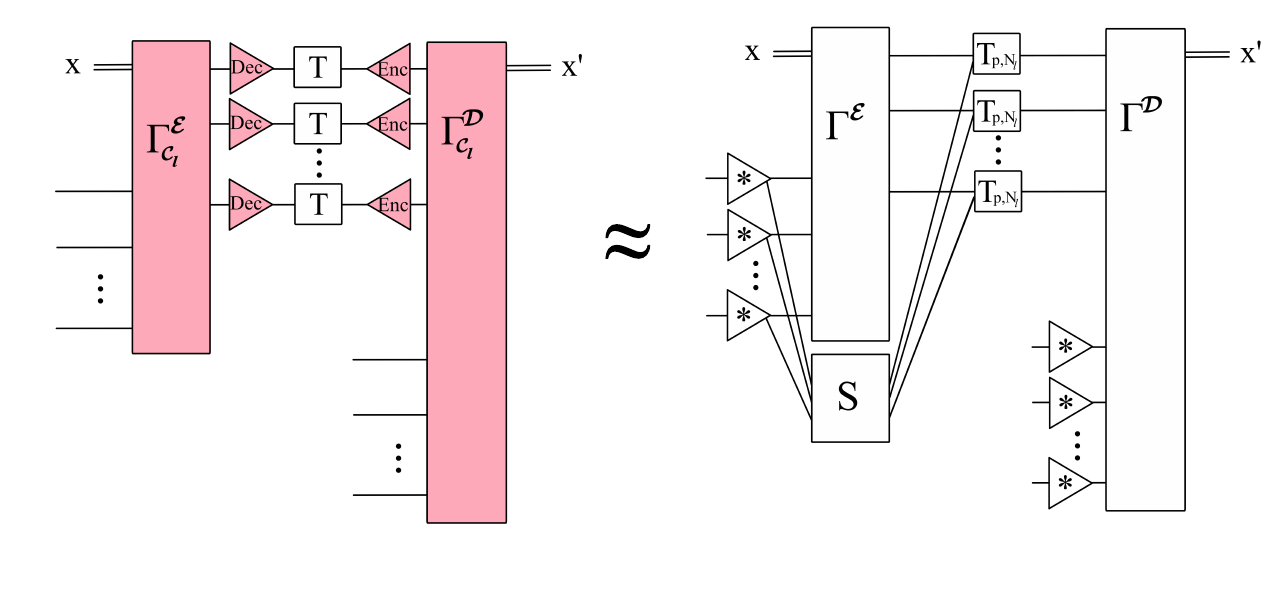}
\caption{\textbf{Sketch of the setup for the effective channel.} The fault-tolerantly implemented encoder $\Gamma_{\mathcal{C}_l}^{\mathcal{E}}$ takes input in the form of $m$ classical bits $x$ and $r$ physical qubits which are encoded in the code space. The resulting codewords are sent through $n$ copies of the quantum channel $T$, preceeded by the decoding interface. The output of the channel is fed into the encoding interface, whose output serves as input to the fault-tolerantly implemented decoder $\Gamma_{\mathcal{C}_l}^{\mathcal{D}}$, which also receives additional quantum input in the form of $s$ qubits in the code space. Theorem~\ref{thm-eff-channel} shows that this setup is very close to a faultless setup with an effective channel, where the quantum systems are transformed by the perfect decoding operation $\DecI^*$, represented by a triangle marked with a star. The effective channel $T_{p,N_l}$ receives input in the form of data qubits and a potentially correlated syndrome state.
}
\label{fig-eff-channel}
\end{figure*}

Lemma~\ref{thm-effective-encoder} and Lemma~\ref{thm-eff-decoder} can be combined to obtain Theorem~\ref{thm-eff-channel}, which is a modified version of Theorem III.9 from \cite{CMH20} that we will use in our analysis of entanglement-assisted capacity.
This theorem links the fault-affected scenario to a communication problem with faultless encoder and decoder circuits connected by an effective noisy channel of a special form, as illustrated in Figure~\ref{fig-eff-channel}.

It is important to note that our setup for fault-tolerant communication considers the operation of encoding information into a quantum channel $T$ and subsequent decoding of this information as two fault-affected circuits connected by $T$. The channel $T$ itself can be taken to model a noisy communication channel, however, we do not consider it as a noise-affected circuit with well-defined fault-locations (hence its white color in the figures). In particular, the noise affecting $T$ can be very different from the noise affecting the encoding and decoding circuits, and does not have to be below threshold. 

\begin{theorem}[Effective channel with quantum input] \label{thm-eff-channel}
Let $T:\mathcal{M}_2^{\otimes j_1} \rightarrow \mathcal{M}_2^{\otimes j_2}$ be a quantum channel, and let $\Enc: \mathbbm{C}^{2^m} \otimes \mathcal{M}_2^{\otimes r}\rightarrow \mathcal{M}_2^{\otimes nj_1}$ be a quantum circuit with $m$ bits of classical input and $r$ qubits of quantum input and let $\Dec:\mathcal{M}_2^{\otimes nj_2} \otimes \mathcal{M}_2^{\otimes s} \rightarrow \mathbbm{C}^{2^m}$ be a quantum circuit with classical output of $m$ bits. For each $l\in\mathbbm{N}$, let $\mathcal{C}_l$ denote the $l$-th level of the concatenated 7-qubit Steane code with threshold $0<p_0\leq 1$. Let $\EncI_l: \mathcal{M}_2 \rightarrow  \mathcal{M}_2^{\otimes 7^l}$ and $\DecI_l: \mathcal{M}_2^{\otimes 7^l} \rightarrow \mathcal{M}_2$ be the interface circuits for the $l$-th level of the concatenated 7-qubit Steane code with threshold $p_0$.

Then, for any $l\in\mathbbm{N}$ and any $0\leq p \leq \min\{p_0/2,1/4c\}$, there exists a quantum channel $N_l: \mathcal{M}_2^{\otimes j_1 7^l }\rightarrow \mathcal{M}_2^{\otimes j_2}$ and a quantum channel $S_S: \mathcal{M}_2^{\otimes j_1(7^l-1) }\rightarrow  \mathcal{M}_2^{\otimes j_1(7^l-1) }$ on the syndrome space such that 
\begin{IEEEeqnarray*}{lCl}
\IEEEeqnarraymulticol{3}{l}{\big\| \big[\Gamma_{\mathcal{C}_l}^{\mathcal{D}} \circ \bigg( \Big( \big( \EncI_l^{\otimes nj_2} \circ T^{\otimes n} \circ \DecI_l^{\otimes nj_1} \circ\Gamma_{\mathcal{C}_l}^{\mathcal{E}} \big)\circ \Compactcdots} \\&& \Compactcdots \circ    (\id_{cl}^{\otimes m} \otimes\EC_l^{\otimes r}) \Big)\otimes\EC_l^{\otimes s}
\bigg) \big]_{\mathcal{F}(p)}  \\&  -&({\Dec}\otimes \Tr_S) \circ \Big(\big( T_{p,N_l}^{\otimes n} \circ ({\Enc}\otimes  S_S )\big)\otimes \id_2^{\otimes s} \otimes \Tr_{S'} \Big)  \circ \Compactcdots \\&& \Compactcdots \circ \Big(\id_{cl}^{\otimes m} \otimes (\DecI_l^*\circ [\EC_l]_{\mathcal{F}(p)})^{\otimes (r+s)}\Big)\big\|_{1 \rightarrow 1}  \\&
\leq& 2 p_0 \Big(\frac{p}{p_0}\Big)^{2^l} (|\Loc(\Enc)|+|\Loc(\Dec)|)  \\&&+2p_0|\Loc(\EncI_1)| \Big(\frac{p}{p_0}\Big)^{2^l-1} j_1n
  \end{IEEEeqnarray*}
with
\[T_{p,N_l}=(1-2(j_1+j_2) cp) (T\otimes \Tr_S) + 2(j_1+j_2)cpN_l\]
with $c=p_0 \max \{|\Loc(\EncI_1)|,|\Loc(\DecI_1\circ EC)|\}$.
$S_S$ may depend on $l, \Enc$ and $\DecI_l$, while $N_l$ may depend on $ l, \EncI_l$ and $\DecI_l$.
\end{theorem}

 This theorem is formulated for quantum channels which map from a quantum system composed of $j_1$ qubits to a quantum system composed of $j_2$ qubits because we consider interfaces between qubits. However, any quantum channel can always be embedded into a quantum channel between systems composed of qubits, such that Theorem~\ref{thm-eff-channel} and subsequent results apply to general quantum channels.

\section{Entanglement-assisted communication with faultless or faulty devices}
\label{sec-ea-cap}

\noindent When a sender and a receiver are connected by many copies of a quantum channel $T$ and have access to entanglement, they can use this setup to transmit a classical message via entanglement-assisted communication. Then, one can identify the best possible operations for the sender and receiver to perform in order to maximize their transmission rate. This section includes a short introduction into entanglement-assisted communication in Section~\ref{sec-ea-cap-normally}, which will serve as a basis for the coding scheme in our main result, followed by a description of the setup for fault-tolerant entanglement-assisted communication and our strategy for its analysis in Section~\ref{sec-ft-cap-setup}. 

\subsection{The entanglement-assisted capacity}
\label{sec-ea-cap-normally}

\noindent Using the superdense coding protocol \cite{BW92}, two classical bits can be communicated by sending only one qubit over a noiseless quantum channel assisted by entanglement. It is therefore natural to study a noisy channel's classical capacity with entanglement assistance \cite{BSST02}. 

To model entanglement-assisted classical communication, we therefore consider a scheme with classical input and output, where quantum entanglement is available to the sender and the receiver.
As sketched in Figure~\ref{fig-ea-cap}, the sender encodes a classical message of $m$ bits into a quantum state of $n$ qudits by performing an encoding map $\mathcal{E}$. The resulting quantum state serves as input into the tensor product of $n$ copies of a quantum channel $T$, which is equivalent to $n$ independent uses of a quantum wire modelled by $T$. Then, the transformed quantum state is decoded by the receiver applying a decoding map $\mathcal{D}$ which converts the channel's output back into a bit string of length $m$. The performance of such a scheme can be quantified by the probability that this resulting bit string and the original message are identical, as formalized in Definition~\ref{def-ea-coding-scheme}. 
Because of superdense coding \cite{BW92} and teleportation \cite{BBCJPW93}, the classical entanglement-assisted capacity of a channel is exactly double its quantum entanglement-assisted capacity.

\begin{figure}[!t]
\centering
       \includegraphics[width=0.5\textwidth]{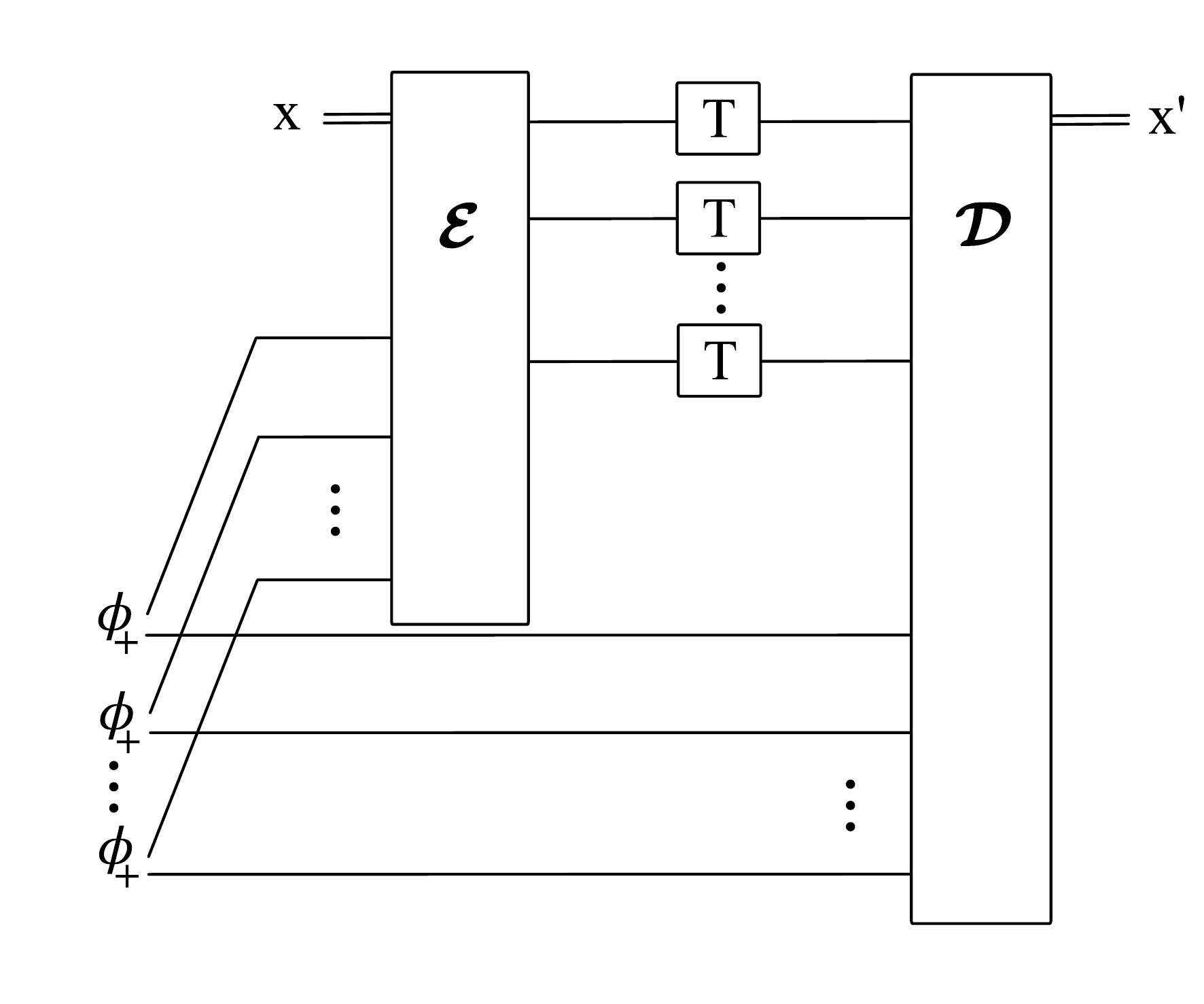}
\caption{\textbf{Basic setup for entanglement-assisted communication.} The encoding map $\mathcal{E}$ maps a bit string $x$ of length $m$ and one part of each entangled state $\varphi$ to a quantum state in $\mathcal{M}_{d_A}^{\otimes n}$. The quantum channel $T$ acts on each of the $n$ subsystems, and the decoder $\mathcal{D}$ uses the other part of each maximally entangled state to decode the received quantum state to a bit string $x'$, which should be identical to the input bit string $x$. Note that classical information transfer is indicated by double lines (input into encoder, output of decoder), while the transfer of quantum states is indicated by single lines.}
\label{fig-ea-cap}
\end{figure}

\begin{definition}[Entanglement-assisted coding scheme]\label{def-ea-coding-scheme}
Let $T:\mathcal{M}_{d_A} \rightarrow \mathcal{M}_{d_B}$ be a quantum channel, and let $n,m \in \mathbbm{N}$, $R_{ea}\in\mathbbm{R}^+$ and $\epsilon>0$.

Then, an $(n,m,\epsilon,R_{ea})$-coding scheme for entanglement-assisted communication consists of quantum channels $\mathcal{E}:\mathbbm{C}^{2^m}\otimes\mathcal{M}_2^{ \otimes \lfloor nR_{ea} \rfloor}\rightarrow\mathcal{M}_{d_A}^{\otimes n}$ and $\mathcal{D}:\mathcal{M}_{d_B}^{\otimes n}\otimes \mathcal{M}_2^{ \otimes \lfloor nR_{ea} \rfloor} \rightarrow \mathbbm{C}^{2^m}$ such that
\begin{align*}
    F\Big(X,\mathcal{D} \circ  \big( (T^{\otimes n } \circ \mathcal{E} )\otimes \id_2^{\otimes \lfloor nR_{ea} \rfloor} \big) (X \otimes \phi_+^{\otimes \lfloor nR_{ea} \rfloor })\Big) \\\geq 1-\epsilon
\end{align*}
where $X=\ketbra{x}$, for all bit strings  $x\in \{0,1\}^m$.
\end{definition}

\begin{remark}
Here, $\phi_+=\ketbra{\phi_+}{\phi_+}$, where $\ket{\phi_+}=\frac{1}{\sqrt{2}}( \ket{00}+\ket{11})$ denotes a maximally entangled state of two qubits. Like in \cite{BSST02}, we define entanglement-assistance with respect to copies of the maximally entangled state $\phi_+$. Without loss of generality, we could allow assistance by copies of arbitrary pure entangled states, since they can be prepared efficiently from maximally entangled states by the process of entanglement dilution \cite{LP99,BSST02}  using of a sublinear amount of classical communication from one party to the other \cite[Theorem~1]{HL04}. It turns out that even entirely arbitrary entangled states (not of product form) cannot increase the communication rate \cite{BDHSW14} and it is therefore sufficient to use maximally entangled states as the entanglement resource.

For a rate of entanglement-assistance $R_{ea}=0$ in the above definition, the scenario reduces to the scheme for classical communication with no entanglement-assistance as introduced in \cite{SW97,Holevo96}. For $R_{ea} \geq  \sup_{\varphi} H(A)_{\varphi}$, where the supremum goes over pure bipartite states $\varphi$, the entanglement-assisted capacity does not increase with more entangled states \cite{HL04}, and we will henceforth focus on this scenario.
Here, $H(A)_{\rho}=- \Tr (\rho \log \rho )$ denotes the von Neumann entropy of a quantum state $\rho \in \mathcal{M}_{A}$.
\end{remark}

 To quantify the difference between the map $\mathcal{D} \circ  \big( (T^{\otimes n } \circ \mathcal{E} )\otimes \id_2^{\otimes nR_{ea}} \big) (\cdot \otimes \phi_+^{\otimes nR_{ea}} )$, corresponding to the coding scheme, and an identity map on $m$ classical bits, corresponding to perfect communication, different measures of distance may be used; here, we use the fidelity $F(\rho,\sigma):=(\Tr (\sqrt{\sqrt{\rho}\sigma \sqrt{\rho}}))^2$ of quantum states $\rho$ and $\sigma$. If the fidelity approaches $1$ in the asymptotic limit, we can communicate with vanishing error, and the best possible rate of communication defines the channel's entanglement-assisted capacity.

\begin{definition}[Entanglement-assisted capacity]
Let $T:\mathcal{M}_{d_A} \rightarrow \mathcal{M}_{d_B}$ be a quantum channel and let $R_{ea}\in \mathbbm{R}_+$ be the rate of entanglement-assistance.

If, for some $R_{ea}$ and for every $n\in \mathbbm{N}$, there exists an $(n,m(n),\epsilon(n),R_{ea})$-coding scheme for entanglement-assisted communication, then a rate $R\geq 0$ is called achievable for entanglement-assisted communication via the quantum channel $T$ if

\begin{IEEEeqnarray*}{ll}
   R & \leq 
   \liminf_{n\rightarrow \infty} \Big\{ 
    \frac{m(n)}{n} \Big\}
    \end{IEEEeqnarray*}
and
\begin{IEEEeqnarray*}{ll}
   \lim_{n\rightarrow \infty} \epsilon(n) \rightarrow 0
    \end{IEEEeqnarray*}
The entanglement-assisted capacity of $T$ is given by
\begin{IEEEeqnarray*}{ll}
  C^{ea}(T) = \sup \{ R & |  R \text{ achievable for entanglement-assisted} \\ &\text{communication via $T$} \}.
    \end{IEEEeqnarray*}
\end{definition}
We will need a more explicit characterization of the communication error that can be reached by using certain entanglement-assisted coding schemes achieving rates close to capacity. The specific bound can be obtained from \cite[Section~20.4]{Wilde13}, and its error term follows from the packing lemma \cite{HDW08}, notions from weak typicality, and Hoeffding's bound \cite{Hoeffding63}.

\subsection{The fault-tolerant entanglement-assisted capacity}
\label{sec-ft-cap-setup}

\noindent In Section~\ref{sec-ea-cap-normally}, the encoder and decoder are assumed to be ideal quantum channels. In order to perform these channels on some given quantum device, they have to be implemented by quantum circuits, i.e., compositions of finitely many elementary gates.
It is well known that quantum devices (unlike classical computers) are notoriously susceptible to faults at the single-gate level which can have devastating effects on the whole computation. This is also true for the circuits encoding and decoding the information that we want to send between different devices or computers. Through clever and protective implementation, the computation within the encoding and decoding devices can be made robust against such faults, raising the question of a channel's fault-tolerant entanglement-assisted capacity.

A coding scheme for a setup affected by noise is defined as follows:

\begin{definition}[Fault-tolerant entanglement-assisted coding scheme]\label{defn:FTEACS}
Let $T:\mathcal{M}_{d_A} \rightarrow \mathcal{M}_{d_B}$ be a quantum channel, and let $n,m \in \mathbbm{N}$, $R_{ea}\in\mathbbm{R}^+$ and $\epsilon>0$.
For $0\leq p \leq 1$, let $\mathcal{F}(p)$ denote the i.i.d. Pauli noise model.

Then, an $(n,m,\epsilon,R_{ea})$-coding scheme for fault-tolerant entanglement-assisted communication consists of quantum circuits $\mathcal{E}:\mathbbm{C}^{2^m}\otimes\mathcal{M}_2^{ \otimes  \lfloor nR_{ea}\rfloor }\rightarrow\mathcal{M}_{d_A}^{\otimes n}$ and $\mathcal{D}:\mathcal{M}_{d_A}^{\otimes n}\otimes \mathcal{M}_2^{ \otimes \lfloor nR_{ea}\rfloor } \rightarrow\mathbbm{C}^{2^m}$ such that
\begin{align*}
F\big(X,\big[\mathcal{D} \circ  \big( (T^{\otimes n} \circ \mathcal{E}\big) \otimes \id_2^{*\otimes  \lfloor nR_{ea}\rfloor }  \big)  \big]_{\mathcal{F}(p)} (X \otimes \phi_+^{\otimes  \lfloor nR_{ea}\rfloor })\big) \\ \geq 1-\epsilon\end{align*}
where $X=\ketbra{x}$, for all bit strings  $x\in \{0,1\}^m$.
\end{definition}

\begin{remark} We emphasize that we assume that the entangled states become subject to faults at the moment where the first gate acts on them. If the entanglement resource in our setup was arbitrary, we could consider a setup assisted by an arbitrary pure entangled state in the code space.  In this scenario, the entanglement resource would directly be available to the fault-tolerantly implemented encoding and decoding circuits, without being corrupted by noisy encoding interfaces, and without necessitating the additional step of entanglement distillation. Achievable rates for such a scenario can be inferred from the expression from our main result (Theorem~\ref{thm-final-coding-thm}) with $f_1(p)=0$.
Under this model, we would assume that the state was prepared and stored (for the duration of the encoding and decoding circuits) within the code space without incurring any faults through the preparation gates and time steps. Here, we choose to consider the more practically relevant scenario where the entanglement resource becomes subject to noise as soon as it enters the code space of the encoding and decoding circuits. The form and amount of entanglement-assistance we consider is as in the standard setup, given by $nR_{ea}$ copies of physical maximally entangled qubits. This scenario also covers situations where $\phi_+$ may be prepared and stored in some highly noise-tolerant and well-suited way until it is needed for computation.
An alternative assumption would be that the maximally entangled states are prepared first and brought into the code space by the encoding interface. Under this assumption, one part of the maximally entangled state (that is to be input into the decoding circuit) waits in the code space for the duration of the encoding circuit and the data transfer via $T$. Instead of $\id_2^*$, a large number identity gates would be performed and add to the overall count of fault-locations. This would not significantly alter our results in Theorem~\ref{thm-final-coding-thm}, but would require a higher choice of the concatenation level $l$ in Eq.~\eqref{eq-level-choice}. Here, we do not consider this scenario in order to simplify the presentation.
\end{remark}

The asymptotically best possible fault-tolerantly achievable rate defines the channel's fault-tolerant entanglement-assisted capacity, fundamentally characterizing how much information the channel can transmit under this noise model.

\begin{definition}[Fault-tolerant entanglement-assisted capacity]
Let $T:\mathcal{M}_{d_A} \rightarrow \mathcal{M}_{d_B}$ be a quantum channel, and let $R_{ea}\in \mathbbm{R}_+$ be the rate of entanglement-assistance. 
For $0\leq p \leq 1$, let $\mathcal{F}(p)$ denote the i.i.d. Pauli noise model.

If, for some $R_{ea}$ and for every $n\in \mathbbm{N}$, there exists an $(n,m(n),\epsilon(n),R_{ea})$-coding scheme for fault-tolerant entanglement-assisted communication under the noise model $\mathcal{F}(p)$, then a rate $R\geq 0$ is called achievable for fault-tolerant entanglement-assisted communication via the quantum channel $T$ if

\begin{IEEEeqnarray*}{ll}
   R & \leq 
   \liminf_{n\rightarrow \infty} \Big\{ 
    \frac{m(n)}{n} \Big\}
    \end{IEEEeqnarray*}
and
\begin{IEEEeqnarray*}{ll}
   \lim_{n\rightarrow \infty} \epsilon(n) \rightarrow 0
    \end{IEEEeqnarray*}
The fault-tolerant entanglement-assisted capacity of $T$ is given by
 \begin{IEEEeqnarray*}{ll}
   C^{ea}_{\mathcal{F}(p)}(T) =\sup &\{R  | R \text{ achievable for fault-tolerant} \\ &\text{entanglement-assisted communication via $T$}\}.
   \end{IEEEeqnarray*}
\end{definition}

\begin{figure}[htbp]
  \centering
       \includegraphics[width=0.5\textwidth]{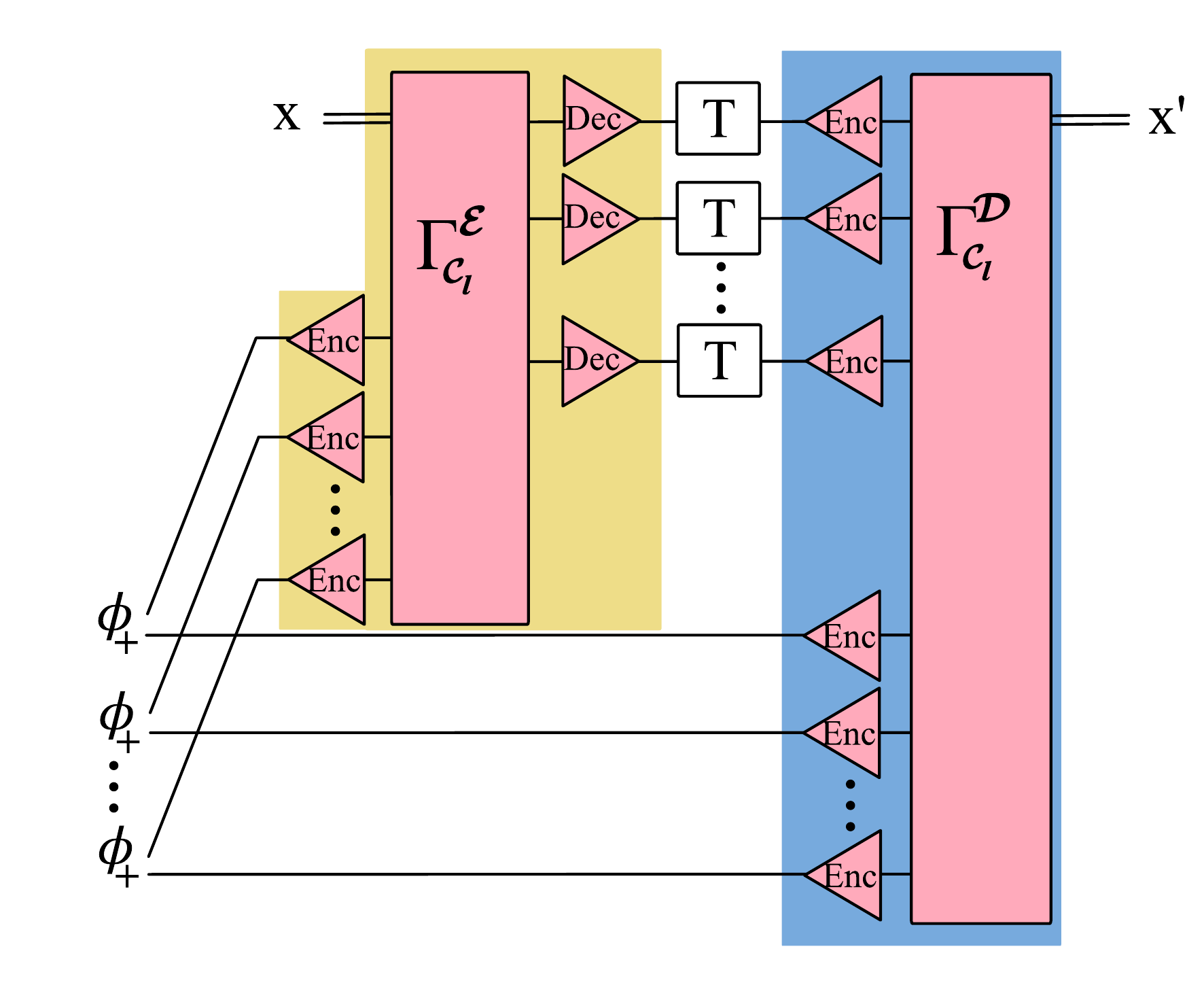}
\caption{\textbf{Basic setup for our coding scheme for fault-tolerant entanglement-assisted communication, see Definition~\ref{defn:FTEACS}.} The encoding map $\mathcal{E}$ (yellow) encodes a bit string of length $m$ into a quantum state that serves as input into $n$ copies of the quantum channel $T$, and the decoding map $\mathcal{D}$ (blue) decodes the received quantum state back to a bit string. In the entanglement-assisted scenario, $\mathcal{E}$ and $\mathcal{D}$ are connected by $\sim nR_{ea}$ maximally entangled states. To make the communication fault-tolerant, the encoding and decoding circuits are implemented fault-tolerantly in an error correcting code $\mathcal{C}_l$ as $\Gamma_{\mathcal{C}_l}^{\mathcal{E}}$ and $\Gamma_{\mathcal{C}_l}^{\mathcal{D}}$, and combined with interfaces $\EncI$ and $\DecI$ mapping between the quantum states serving as input and output for $T$, and the quantum states being transformed in the fault-tolerantly implemented encoding and decoding circuits.
}
\label{fig-ft-ea-cap-setup}
\end{figure}

Formally, any quantum circuits $\mathcal{E}$ and $\mathcal{D}$ may be chosen in Definition~\ref{defn:FTEACS}, leading to a coding scheme for fault-tolerant entanglement-assisted classical communication. To prove lower bounds to the fault-tolerant capacity $C^{ea}_{\mathcal{F}(p)}(T)$ for a quantum channel $T$ in terms of the capacity $C^{ea}(T)$, we will use a particular construction that is similar to constructions in \cite{CMH20}. 

Consider some coding scheme not affected by noise for entanglement-assisted classical communication over the channel $T$. We can turn this coding scheme into a fault-tolerant coding scheme by first approximating it by quantum circuits, and then implementing these quantum circuits in a high level of the concatenated $7$-qubit Steane code. Crucially, we will use the interface circuits from \cite{MGHHLPP14,CMH20} to convert between physical qubits and logical qubits in the code space, e.g., when qubits from the output of the channel $T$ are brought into the code space. Unfortunately, these interfaces fail with a probability $2cp$, where $p$ is the gate error probability of the noise model and $c$ some interface-dependent constant (from Theorem~\ref{thm-correct-interfaces}) and the fault-tolerant implementation of the coding scheme affected by faults will not be equivalent to the original coding scheme for the quantum channel $T$. Instead it will be equivalent to the coding scheme for a certain effective quantum channel $T_{p,N}$ as in Theorem~\ref{thm-eff-channel}. 

Our strategy starts by considering a coding scheme for entanglement-assisted classical communication for channels that include our effective communication channels $T_{p,N}$. We refer to this channel model as \emph{arbitrarily varying perturbation} (AVP) and we will discuss it in detail in Section~\ref{sec-avp}. This model has been introduced in \cite{CMH20} in the cases of unassisted classical and quantum communication, and it is closely related to the fully-quantum arbitrarily varying channels studied in \cite{BDNW18}. As described in the preceeding paragraphs, we then obtain a fault-tolerant coding scheme by implementing the coding scheme under AVP in a high level of the concatenated $7$-qubit Steane code. For the fault-tolerant entanglement-assisted capacity, the setup crucially includes a supply of maximally entangled states that are connected to the fault-tolerantly implemented encoder and decoder circuit via additional interfaces, as illustrated in Figure~\ref{fig-ft-ea-cap-setup}. Because of the effective probability of failure of these interfaces, when transferring the maximally entangled states into the code space, they are only correctly transmitted with a probability of approximately $1-4cp$ (since there is one interface for each qubit). Subsequently, the entanglement inserted into the code space is noisy and in a mixed state. To counteract this, we show that this entanglement can still be made usable by transforming it back into pure state entanglement in the code space by performing (fault-tolerant) entanglement distillation in Section~\ref{sec-ft-ent-dist}. Since entanglement distillation requires classical communication, we will need to use a subset of the channels $T$ to run the fault-tolerant protocol from \cite{CMH20} to send classical information. Thereafter, with slightly fewer copies of $T$ remaining, an analysis similar to \cite{CMH20} is carried out in Section~\ref{sec-coding-thm} to arrive at a coding theorem describing rates of fault-tolerant entanglement-assisted communication that are achievable with vanishing error.

\section{Entanglement-assisted communication under arbitrarily varying perturbation}
\label{sec-avp}

\noindent As described in Section~\ref{sec-ft-cap-setup}, we find a correspondence between the capacity of a fault-affected setup and an information-theoretic communication setup under non-i.i.d. perturbations which we outline in Section~\ref{sec-avp-model}. Based on similar channel models in \cite{CMH20}, we introduce a generalized version of an entanglement-assisted capacity which allows for arbitrarily varying syndrome input and prove a coding theorem for this model in Section~\ref{sec-avp-coding-thm}.

\subsection{The entanglement-assisted capacity under arbitrarily varying perturbation}\label{sec-avp-model}

\noindent One key feature of the communication problem emerging from Theorem~\ref{thm-eff-channel} is that the effective channel takes input from the space of channel coding symbols as well as the syndrome space. Since the syndrome state can be correlated across different channel uses, the effective communication problem is not covered by standard i.i.d. communication scenarios, but instead defines a communication problem of its own as introduced in \cite{CMH20} and with similarities to communication scenarios studied in \cite{BDNW18}.
 
For $T:\mathcal{M}_{d_A}\rightarrow \mathcal{M}_{d_B}$ and any quantum channel $N:\mathcal{M}_{d_A}\otimes \mathcal{M}_{d_S} \rightarrow \mathcal{M}_{d_B}$ let  $T_{p,N}:\mathcal{M}_{d_A}\otimes \mathcal{M}_{d_S} \rightarrow \mathcal{M}_{d_B}$ denote the quantum channel $T_{p,N}=(1-p) (T\otimes \Tr_S) + pN $. Here, we consider the problem of communicating via a channel of the form
$T_{p,N}^{\otimes n} (\cdot\otimes \sigma_S):\mathcal{M}_{d_A}^{\otimes n} \rightarrow \mathcal{M}_{d_B}^{\otimes n}$ for arbitrary syndrome states $\sigma_S$ and arbitrary quantum channels $N$. We refer to this model as \textit{communication under arbitrarily varying perturbation} (AVP). 

Note that the effective channel model emerging from Theorem~\ref{thm-eff-channel} takes the form $T_{p,N}^{\otimes n} (\cdot\otimes \sigma_S)$ for some syndrome state $\sigma_S$ and some quantum channels $N$, where the dimension $d_S$ depends on the level of concatenation. Since the level of concatenation has to increase with the number $n$ of channel uses if we want the overall communication error to vanish in a fault-tolerant communication scenario, we will allow $d_S$ to be arbitrary and possibly dependent on $n$ in the definition of the capacity under AVP. Previous results in \cite{BDNW18} for models with fixed syndrome state dimension are therefore not directly applicable in this setting. Notably, $\sigma_S$ can be arbitrarily entangled between the $n$ spaces. If $\sigma_S$ is a separable state, this communication model is a special case of a channel model studied in \cite{ABBN12}. Our definitions here consider general syndrome states and can thereby be taken to apply to the worst-case scenario with arbitrarily correlated syndrome states.

\begin{definition}[Entanglement-assisted coding scheme under arbitrarily varying perturbation]
\label{def-avp-coding-scheme}
Let $T:\mathcal{M}_{d_A} \rightarrow \mathcal{M}_{d_B}$ be a quantum channel, and let $n,m \in \mathbbm{N}$, $R_{ea}\in\mathbbm{R}^+$ and $\epsilon>0$.

Then, an $(n,m,\epsilon,R_{ea})$-coding scheme for entanglement-assisted communication under AVP of strength $p$ consists of
a pure bipartite quantum state $\varphi\in \mathcal{M}_2^{\otimes 2}$ and the quantum channels $\mathcal{E}:\mathbbm{C}^{2^m}\otimes\mathcal{M}_2^{ \otimes nR_{ea}}\rightarrow\mathcal{M}_{d_A}^{\otimes n}$ and $\mathcal{D}:\mathcal{M}_{d_A}^{\otimes n}\otimes \mathcal{M}_2^{ \otimes nR_{ea}} \rightarrow\mathbbm{C}^{2^m}$ such that
\begin{align*} \inf F\Big(X, \mathcal{D} \circ  \big( (T_{p,N}^{\otimes n}  \circ (\mathcal{E} \otimes \sigma_S) ) \otimes \id_2^{\otimes nR_{ea}}  \big) (X \otimes \varphi^{\otimes nR_{ea}}) \Big) \\\geq 1-\epsilon\end{align*}
where $X=\ketbra{x}$, for all bit strings  $x\in \{0,1\}^m$, where
\[T_{p,N}= (1-p) (T\otimes \Tr_S ) + pN . \]
The infimum goes over the dimension $d_S\in \mathbbm{N}$, quantum states $\sigma_S\in\mathcal{M}_{d_S}^{\otimes n}$, and quantum channels $N:\mathcal{M}_{d_A}\otimes \mathcal{M}_{d_S} \rightarrow \mathcal{M}_{d_B}$.
\end{definition}

\begin{remark}
Here, we consider copies of arbitrary bipartite pure entangled states instead of maximally entangled states as entanglement resource, as the latter would require an extra step of entanglement dilution for the coding scheme from \cite[Theorem~21.4.1]{Wilde13}. In order to ensure exponential decay of the entanglement dilution error that we require in our proof, the communication rate would be reduced by some function linear in perturbation strength $p$. However, this is not necessary in the context of fault-tolerant coding because the extra classical communication can be performed together with the entanglement distillation, leading to an overall better bound on the achievable rate in our main result, Theorem \ref{thm-final-coding-thm}.
\end{remark}

Naturally, the best possible rate of information transfer using such a coding scheme is a channel's entanglement-assisted capacity under AVP:

\begin{definition}[Entanglement-assisted capacity under arbitrarily varying perturbation]
Let $T:\mathcal{M}_{d_A} \rightarrow \mathcal{M}_{d_B}$ be a quantum channel, and let $R_{ea}\in\mathbbm{R}^+$.

If, for some $R_{ea}$ and for every $n\in \mathbbm{N}$, there exists an $(n,m(n),\epsilon(n),R_{ea})$-coding scheme for entanglement-assisted communication under AVP of strength $p$, then a rate $R\geq 0$ is called achievable for entanglement-assisted communication under AVP via the quantum channel $T$ if 
\begin{IEEEeqnarray*}{ll}
   R & \leq 
   \liminf_{n\rightarrow \infty} \Big\{ 
    \frac{m(n)}{n} \Big\}
    \end{IEEEeqnarray*}
and
\begin{IEEEeqnarray*}{ll}
   \lim_{n\rightarrow \infty} \epsilon(n) \rightarrow 0
    \end{IEEEeqnarray*}
The entanglement-assisted capacity of $T$ under AVP is given by
\begin{IEEEeqnarray*}{ll}
   C^{ea}_{AVP}(p,T) =\sup \{R &| R \text{ achievable for entanglement-assisted} \\ &\text{communication under AVP via $T$}\}.
   \end{IEEEeqnarray*}
\end{definition}

This version of entanglement-assisted capacity thus characterizes how well one can communicate with non-i.i.d. channel input, which may be interesting for various communication problems, and appears in particular in our study of fault-tolerant communication in Section~\ref{sec-coding-thm}.

\subsection{A coding theorem for entanglement-assisted communication under arbitrarily varying perturbation} \label{sec-avp-coding-thm}

\noindent The information-theoretic model outlined in Section~\ref{sec-avp-model} naturally raises the question of how much such AVP can hinder communication. Here, we show that communication is still possible in this scenario, and that achievable rates are given by the following theorem:

\begin{theorem}[Lower bound on the entanglement-assisted capacity under AVP] \label{thm-avp-coding-thm}
For any quantum channel $T:\mathcal{M}_{d_A} \rightarrow \mathcal{M}_{d_B}$, and for any $0\leq p \leq 1$, we have an entanglement-assisted capacity under AVP with 
\[C^{ea}_{AVP}(p,T) \geq  C^{ea}(T)-g(p) \]
where
\begin{IEEEeqnarray*}{lCl}
g(p)& =&
2(d_Ad_B\log(d_Ad_B) +1)\sqrt{2{\log(d_B)p}} |\log( \frac{p^2}{d_Ad_B} )|\\&& +  2h\Big(d_Ad_B\sqrt{2{\log(d_B)p}} |\log(\frac{p^2}{d_Ad_B})|\Big) \\&& 
+5p\log(d_B)+ 2(1+2p)h\Big(\frac{2p}{1+2p}\Big) 
\\&=&\mathcal{O}(p\log(p)).
\end{IEEEeqnarray*}
\end{theorem}
\begin{IEEEproof}[Proof of Theorem~\ref{thm-avp-coding-thm}]
In this proof, we find a coding scheme for entanglement-assisted communication under AVP with strength $p\in\left[ 0,1\right]$ by constructing an encoder channel $\mathcal{E}$ and a decoder channel $\mathcal{D}$. Achievable rates are rates for which the fidelity in Definition~\ref{def-avp-coding-scheme} goes to $1$, which corresponds to rates $R$ for which the following expression goes to zero (which is a consequence of Fuchs-van-de-Graaf-inequality \cite{FvdG97}):

\begin{align}\begin{split}
    \label{this-has-to-go-to-zero-avp}
\sup&\{
\| \id_{cl}^{\otimes nR} - \mathcal{D} \circ  (T_{p,N}^{\otimes n} \circ (\mathcal{E} \otimes \sigma_S) \otimes  \id_2^{\otimes nR_{ea}}) \circ \Compactcdots \\& \Compactcdots \circ   (\id_{cl}^{\otimes nR} \otimes \varphi^{\otimes nR_{ea}}) \|_{1 \rightarrow 1}  \} 
 \overset{n\rightarrow \infty}{\rightarrow}  0 \end{split}
\end{align}
where the supremum goes over the dimension $d_S\in \mathbbm{N}$, quantum states $\sigma_S\in\mathcal{M}_{d_S}^{\otimes n}$, and quantum channels $N:\mathcal{M}_{d_A}\otimes \mathcal{M}_{d_S} \rightarrow \mathcal{M}_{d_B}$.

Our construction makes use of a particular coding scheme for entanglement-assisted communication: For any quantum channel $T$, we consider the quantum channel $T_p=(1-p) T + p \frac{\mathbbm{1}}{d_B} \Tr(\cdot )$ with $p$ being the strength of the AVP. Using the coding scheme from \cite[Theorem~21.4.1]{Wilde13} (see also Appendix~\ref{appendix-codingerror}), for any quantum channel $T_p$, and any pure bipartite quantum state $\varphi\in \mathcal{M}_{d_A}\otimes \mathcal{M}_{d_A}$, there exists an encoder $\mathcal{E}$ and a decoder $\mathcal{D}$ for $T_p$ such that
\begin{align}\begin{split}  \label{eq-fid} &\Xi( T_p^{\otimes n })\\&:= F\Big(X,\mathcal{D} \circ  \big( ( T_p^{\otimes n } \circ \mathcal{E} ) \otimes \id_{2}^{\otimes nR_{ea}} \big) (X \otimes \varphi^{\otimes nR_{ea} })\Big)  \\&\geq 1-{\epsilon_{ea}} \end{split}\end{align}
for any classical message $x$ with the corresponding quantum state $X=\ketbra{x}$ of length $nR'$, where 
\begin{IEEEeqnarray}{lCl}\label{eq-coding-error} \epsilon_{ea}&\leq &
12 e^{-\frac{n\delta^2}{2(\log(\lambda_{\min}))^2}}\\&&  + 16 \cdot 2^{-n(I(A':B)_{(T_p\otimes \id_2)(\varphi)} -\eta(\delta,d_A,d_B)-d_Ad_B\frac{\log(n+1)}{n} -R')}\IEEEnonumber\end{IEEEeqnarray}
with the function $\eta(\delta,d_A,d_B)=2 (d_Ad_B\log(d_Ad_B) +4)  \delta+2h(d_Ad_B\delta)$ and with the smallest non-vanishing eigenvalue $\lambda_{\min}=\min\{\lambda\in \Spec((T_p\otimes\id_2)(\varphi))| \lambda>0 \}$. Here, $h(x)=-x\log(x)-(1-x)\log(1-x)$ denotes the binary entropy and $I(A:B)_{\rho}=H(A)_{\rho}+H(B)_{\rho}-H(AB)_{\rho}$ denotes the quantum mutual information.

To apply this coding scheme to the original channel $T$ under AVP, we will use the postselection-type result in \cite[Lemma~IV.10]{CMH20}: For any $\tilde{\delta}>0$, we have
\[ T_{p,N}^{\otimes n} (\cdot \otimes \sigma_S) \leq d_B^{n(p+\tilde{\delta})} T_p^{\otimes n} + e^{-\frac{n\tilde{\delta}^2}{3p}} S \]
for some quantum channel $S:\mathcal{M}_{d_A}^{\otimes n}\rightarrow \mathcal{M}_{d_B}^{\otimes n}  $. Here, we write $S_1\leq S_2$ for completely positive maps $S_1$ and $S_2$ if the difference $S_2-S_1$ is completely positive. Using a simple monotonicity property of the fidelity, we have: 
\begin{IEEEeqnarray*}{Cl}
 \IEEEeqnarraymulticol{2}{l}{\|\id_{cl}^{\otimes nR}- \Big( \mathcal{D} \circ (T_{p,N}^{\otimes n} \circ (\mathcal{E} \otimes \sigma_S) \otimes  \id_2^{\otimes nR_{ea}}) \circ \Compactcdots} \\& \Compactcdots \circ  (\id_{cl}^{\otimes nR} \otimes {\varphi}^{\otimes  nR_{ea}} ) \Big) \|_{1 \rightarrow 1}  \\
 \leq& 2 \sqrt{1-\Xi( {T}_{p,N}^{\otimes n }) }\\
\leq & 2 \sqrt{ d_B^{(p+\tilde{\delta})n} \Big(1- \Xi(  {T}_p^{\otimes n }) \Big)- e^{-\frac{n \tilde{\delta}^2}{3p}} } \\
\leq & 2\sqrt{d_B^{(p+\tilde{\delta})n} \epsilon_{ea}-e^{-\frac{n \tilde{\delta}^2}{3p}}  }\\
\end{IEEEeqnarray*}
where we make use of \cite[Proposition~4.3]{KW04} for the first inequality, \cite[Lemma~IV.10]{CMH20} for the second inequality, and equation Eq.~\eqref{eq-fid} in the last inequality. Clearly, we have Eq.~\eqref{this-has-to-go-to-zero-avp} if

\begin{IEEEeqnarray*}{ll} d_B^{(p+\tilde{\delta})n} \epsilon_{ea} & \overset{n\rightarrow \infty}{\rightarrow}  0\end{IEEEeqnarray*}
with $\epsilon_{ea}$ from Eq.~\eqref{eq-coding-error}.


For any  $\tilde{\delta}>0$ sufficiently small, we thus obtain a bound on the choices of $\delta$ and $R$, where the choice of $\delta$ should guarantee that
\begin{equation}\label{eq:deltaErrorBound}
d_B^{(p+\tilde{\delta})n} e^{-\frac{n\delta^2}{2(\log(\lambda_{\min}))^2}}\rightarrow 0 ,
\end{equation}
while the bound on $R$ should guarantee that
\begin{align}\begin{split} \label{eq:restErrorBound}
d_B^{(p+\tilde{\delta})n}  2^{-n(I(A':B)_{(T_p\otimes \id_2)(\varphi)} -\eta(\delta,d_A,d_B)-\frac{d_Ad_B}{n}\log(n+1) -R)}\\ \overset{n\rightarrow \infty}{\rightarrow} 0.\end{split}
\end{align}
To guarantee that Eq.~\eqref{eq:deltaErrorBound} holds, we choose $\delta$ as
\[\delta =\sqrt{2{\log(d_B)p}} |\log(\lambda_{min})| ,\]
and $\tilde{\delta}>0$ sufficiently small. We then find the bound
\begin{align}\begin{split} \label{eq-r-avp}
    R <&  I(A':B)_{( T_p \otimes \id_2)(\varphi)} - \log(d_B)p
     \\&-\eta(\sqrt{2{\log(d_B)p}} |\log(\lambda_{min})|,d_A,d_B)
    \end{split}
\end{align}
such that Eq.~\eqref{eq:restErrorBound} holds. Thereby, we obtain
\begin{IEEEeqnarray*}{Cl}
C^{ea}_{AVP}(p,T) \geq& 
I(A':B)_{(T_p \otimes \id_2)(\varphi)}- \log(d_B)p   \\&
-\eta(\sqrt{2{\log(d_B)p}} |\log(\lambda_{min}) |,d_A,d_B)  
\end{IEEEeqnarray*}
for any pure quantum state $\varphi$, where $\lambda_{\min}=\min\{\lambda\in \Spec((T_p\otimes\id_2)(\varphi))| \lambda>0 \} \}$.

To get an expression in terms of the usual entanglement-assisted capacity $C^{ea}(T)$, we use continuity estimates in the following way: Consider the pure quantum state $\varphi^*=\text{argmax}_{\varphi} I(A':B)_{( T\otimes \id_2)(\varphi)}$ which achieves the maximum for the quantum mutual information for the channel $T$.
 Then, consider a pure quantum state $\varphi_p=(1-p)\varphi^* + p\phi_+$.
 For this state, we have that $\|(T_p\otimes \id_2)(\varphi_p)-(T_p\otimes \id_2)(\varphi^*) \|_{\Tr} \leq p$. 
Furthermore, we have that $( T_p\otimes \id_2)(\varphi_p)=(1-p)( T_p\otimes \id_2)(\varphi^*)+p( T_p\otimes \id_2)(\phi_+)=(1-p)( T_p\otimes \id_2)(\varphi^*)+p (1-p) ( T\otimes \id_2)(\phi_+)+ p^2  \frac{\mathbbm{1}}{d_Ad_B}$. Because of this, and because all summands are positive semi-definite, the minimum eigenvalue of $( T_p\otimes \id_2)(\varphi_p)$ is lower bounded as $\lambda_{\min} \geq   \frac{p^2}{d_Ad_B}$. In addition, we know that $ \|T(\rho)-T_p(\rho) \|_{\Tr} \leq p $ for all quantum states $\rho$.
  By triangle inequality, we therefore find
  \[  \|(T_p\otimes \id_2)(\varphi_p)-(T\otimes \id_2)(\varphi^*) \|_{\Tr} \leq 2p \]
 Then, using the continuity of mutual information \cite[Corollary~1]{Shirokov17}, we find that
 \begin{align*}
    & | I(A':B)_{( T_p \otimes \id_2)(\varphi_p)} - I(A':B)_{( T \otimes \id_2)(\varphi*)} | \\
     &\leq 4p\log(d_B)+ 2(1+2p)h\Big(\frac{2p}{1+2p}\Big) 
 \end{align*} 
In total, we thereby obtain the bound \[C^{ea}_{AVP}(p,T) \geq  C^{ea}(T)-g(p) \]
where
\begin{IEEEeqnarray*}{lCl}
g(p)&=&     \eta(\sqrt{2{\log(d_B)p}} |\log(\frac{p^2}{d_Ad_B})|,d_A,d_B) 
\\&&+5p\log(d_B)+ (1+2p)h\Big(\frac{2p}{1+2p}\Big) \\& =&
2(d_Ad_B\log(d_Ad_B) +1)\sqrt{2{\log(d_B)p}} |\log( \frac{p^2}{d_Ad_B} )|\\&& +  2h\Big(d_Ad_B\sqrt{2{\log(d_B)p}} |\log(\frac{p^2}{d_Ad_B})|\Big) \\&& 
+5p\log(d_B)+ 2(1+2p)h\Big(\frac{2p}{1+2p}\Big) 
\end{IEEEeqnarray*}
\end{IEEEproof}

As a consequence of this result, we find the following continuity in perturbation strength $p$ of the entanglement-assisted capacity under AVP. Moreover, the usual notion of entanglement-assisted capacity is recovered for vanishing perturbation probability $p$.

\begin{theorem}
For every $\eta>0$ and $d_A,d_B\in\mathbbm{N}$ there exists a $p(\eta,d_A,d_B)\in \left[ 0,1\right]$ such that 
\[
C^{ea}_{AVP}(p,T) \geq C^{ea}(T) - \eta,
\]
for every $p\leq p(\eta,d_A,d_B)$ and every quantum channel $T:\mathcal{M}_{d_A} \rightarrow \mathcal{M}_{d_B}$.
\end{theorem}

\begin{corollary}
Let $p\geq0$. Then, for every quantum channel $T:\mathcal{M}_{d_A}\rightarrow \mathcal{M}_{d_B}$, we have 
\[\lim_{p\rightarrow 0} C^{ea}_{AVP}(p,T) = C^{ea}(T).\]
\end{corollary}

\section{Fault-tolerant entanglement distillation} \label{sec-ft-ent-dist}

\noindent As outlined in Section~\ref{sec-ea-cap}, the entanglement-assisted capacity considers encoders and decoders as general quantum channels that have access to entanglement.
In a fault-tolerant setup, framing the encoder and decoder as circuits with an implementation in a fault-tolerant code means that the entanglement has to be transferred into the code space through an interface as explained in Section~\ref{sec-ft-cap-setup}. Naturally, this interface is in itself a fault-affected circuit, and can produce a noisy mixed state in the code space. 

\begin{figure}[htbp]
  \centering
       \includegraphics[width=0.5\textwidth]{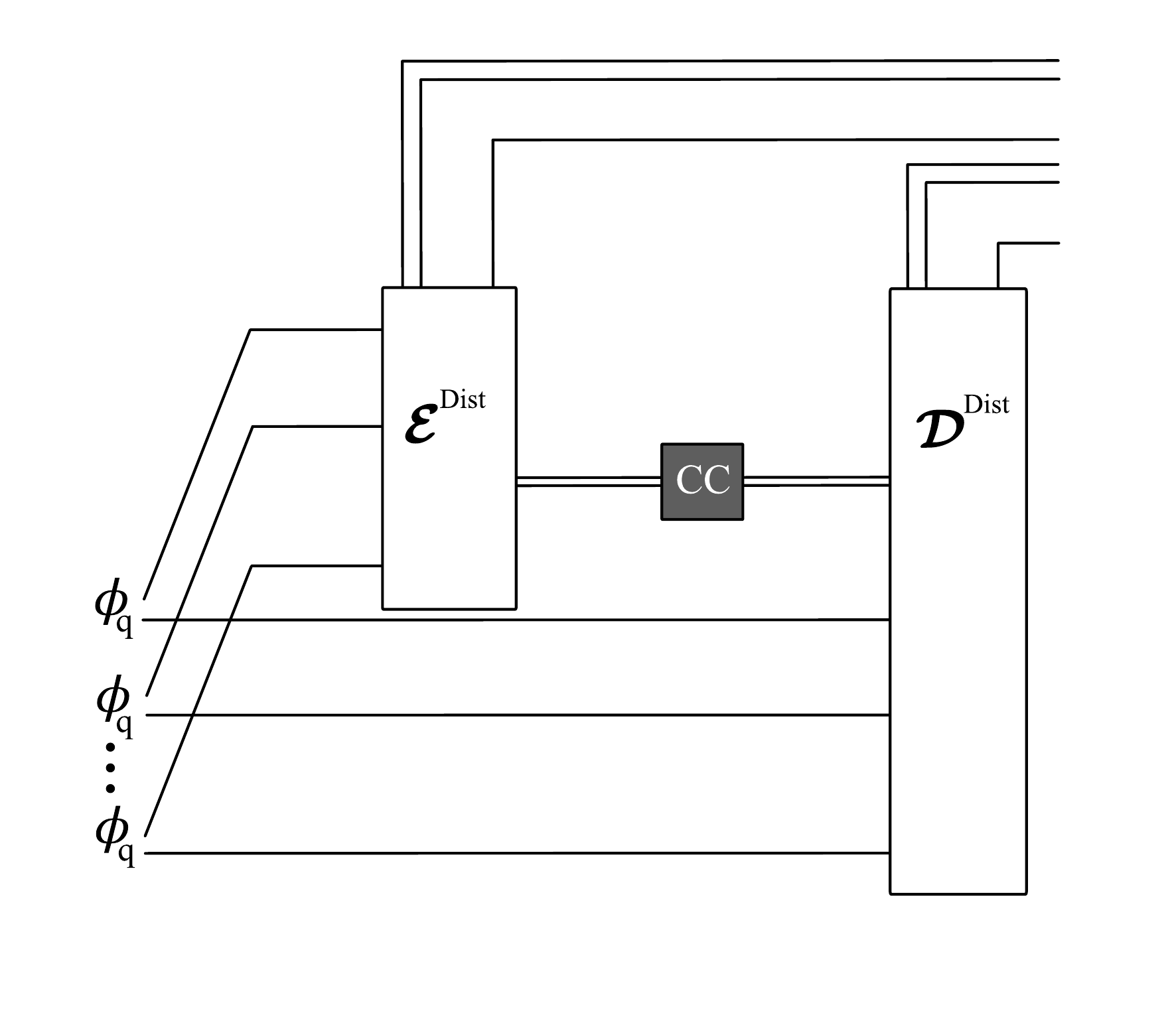}
\caption{\textbf{Setup for entanglement distillation based on the protocol in \cite{DW03}.} Two parties each have access to one part of $k$ noisy entangled states $\phi_q$. One party performs local operations ${\mathcal{E}}^{\Dist}$ and sends one-way classical communication to the other, who performs local operations ${\mathcal{D}}^{\Dist}$. The output state of this scheme is close in fidelity to $(1-H(\phi_q))k$ copies of the maximally entangled state (cf. Theorem~\ref{thm-ent-dist-original}).
}
\label{fig-ent-dist-normally}
\end{figure}

Fortunately, we can asymptotically carry out entanglement distillation with one-way classical communication \cite{DW03}, thereby transforming many copies of a noisy entangled state into fewer copies of a perfectly maximally entangled state.

\begin{theorem}[{Entanglement distillation, \cite[Theorem~10]{DW03}}] \label{thm-ent-dist-original}
Let $\{\phi_+,\phi_-,\psi_+,\psi_-\}$ be the Bell basis of the space $\mathcal{M}_2^{\otimes 2}$.

For sufficiently large $k\in\mathbbm{N}$, there exists a $\delta\geq0$ and a quantum channel $\Dist:\mathcal{M}_2^{\otimes 2k}\rightarrow \mathcal{M}_2^{\otimes 2(1-H(\phi_q)))k}$ consisting of local operations and $H(\phi_q) k$ bits of one-way classical communication, such that $k$ copies of the state $\phi_{q}=(1-q) \phi_+ +\frac{q}{3} (\phi_- + \psi_+ + \psi_-)$ are mapped to $(1-H(\phi_q))k$ copies of the maximally entangled state $\phi_+$
with the following fidelity:
\[ F\big((\phi_+)^{\otimes (1-H(\phi_{q}))k}, \Dist(\phi_q^{\otimes k}) \big) \geq 1- \epsilon_{dist}(q)\]
with 
\[\epsilon_{dist}(q) \leq 
2e^{-\frac{k \delta^2}{\log(q/3)^2}}+\sqrt{2\sqrt{3}e^{-k \frac{\delta^2}{2\log(q/3)^2}}}\]
\end{theorem}

\begin{remark}
The von Neumann entropy of the state $\phi_{q}$ is $ H(\phi_{q})=-(1-q)\log(1-q) - q \log(\frac{q}{3})=h(q)+q\log (3)$ where $h$ denotes the binary entropy. We restrict ourselves to using states of this form because we consider the Pauli i.i.d. fault model, but similar considerations can be made for general noisy input states. In that case, the amount of maximally entangled states that can be obtained from copies of the state $\rho$ is given by its distillable entanglement $H(A)_{\rho}-H(AB)_{\rho}$ per copy, and the amount of classical communication that has to be performed amounts to $H(A)_{\rho}-H(B)_{\rho}+H(R)_{\rho}$ bits per copy, where $R$ denotes a purifying system \cite[Remark~11]{DW03}. In other words, our results extend to fault-tolerant entanglement distillation from arbitrary states, where the noisy state $\rho$ is additionally transformed by the noisy effective interfaces such that the entanglement effectively has to be distilled from the state $(N_{enc,p,l} \otimes N_{enc,p,l})(\rho)$. Fault-tolerant entanglement distillation could still be performed, but may require a higher number of copies of $T$ and may lead to fewer perfect maximally entangled pairs in the code space.
\end{remark}

\begin{remark}
Assuming faultless encoder and decoder circuits, but noisy mixed entangled states, using a subset of the channel copies for entanglement distillation implies a notion of classical capacity with assistance by noisy states, which may be of independent interest. 
For $R_{ea}\geq \frac{\log(2)}{1-H(\phi_q)}$, we see that the capacity $C_{\phi_q}^{ea}$ with assistance by $nR_{ea}$ copies of the state $\phi_{q}=(1-q) \phi_+ +\frac{q}{3} (\phi^- + \psi^+ + \psi^-)$ is given by $C_{\phi_q}^{ea}(T)\geq C_{\phi_+}^{ea}(T)- 
H(\phi_q) R_{ea} \frac{C_{\phi_+}^{ea}(T)}{C(T)} $.
Note that this is related to the scheme of special dense coding, a generalized version of superdense coding where the sender and receiver have access to arbitrary pairs of qubits and are connected by a noiseless, perfect quantum channel. Using purification procedures \cite{BPV98,Bowen01,ZZYS17} or directly finding a coding scheme \cite{HHHLT01}, coding protocols have been proposed and achievable rates have been computed for this scenario, where the latter also shows that states with bounded (i.e. non-distillable) entanglement do not enhance communication via a perfect quantum channel at all.
\end{remark}

Implementing the circuits for the distillation machines fault-tolerantly requires physical states to be inserted into the code space via an interface.
This interface is also subject to the fault model and only correct with a certain probability which cannot be made arbitrarily small by increasing the concatenation level of the concatenated $7$-qubit Steane code. Effectively, this leads to noisy states in the code space, which the distillation tries to counteract.

Note that this means that the input state into the whole protocol, the original sea of maximally entangled states can still be assumed to be noiseless; the input to our protocol for entanglement-assisted communication will be in the form of perfectly maximally entangled physical qubits that become noisy because of the fault-affected interface, as sketched in Figure~\ref{fig-ft-ea-cap-setup}. In summary, this proposed scheme for fault-tolerant distillation takes perfectly maximally entangled physical qubits as an input, and the desired output is in the form of perfectly maximally entangled states in the code space.

In order to show that a fault-tolerant distillation protocol with the concatenated $7$-qubit Steane code can transform physical maximally entangled states such that they are very close to maximally entangled states in the code space, we are going to make use of Lemma~\ref{thm-effective-encoder}. Because of the fault-affected interface, any quantum state that serves as input into a fault-tolerantly implemented circuit is transformed by the interface into an effective input state which is a mixture of the original state (with weight of approximately $1-4cp$, where $c$ is the constant from Theorem~\ref{thm-correct-interfaces}) and a noisy state (with weight of $4cp$), serving as input into the perfect circuit. This is true in particular for $k$ copies of the maximally entangled state. Then, we can employ fault-tolerant circuits implementing the protocol from \cite{DW03} to fault-tolerantly restore maximal entanglement for $(1-h(4cp)-4cp\log(3))k$ qubit states in the code space. Our setup for fault-tolerant distillation is sketched in Figure~\ref{fig-ft-ent-dist}.

\begin{figure}[htbp]
  \centering
      \includegraphics[width=0.5\textwidth]{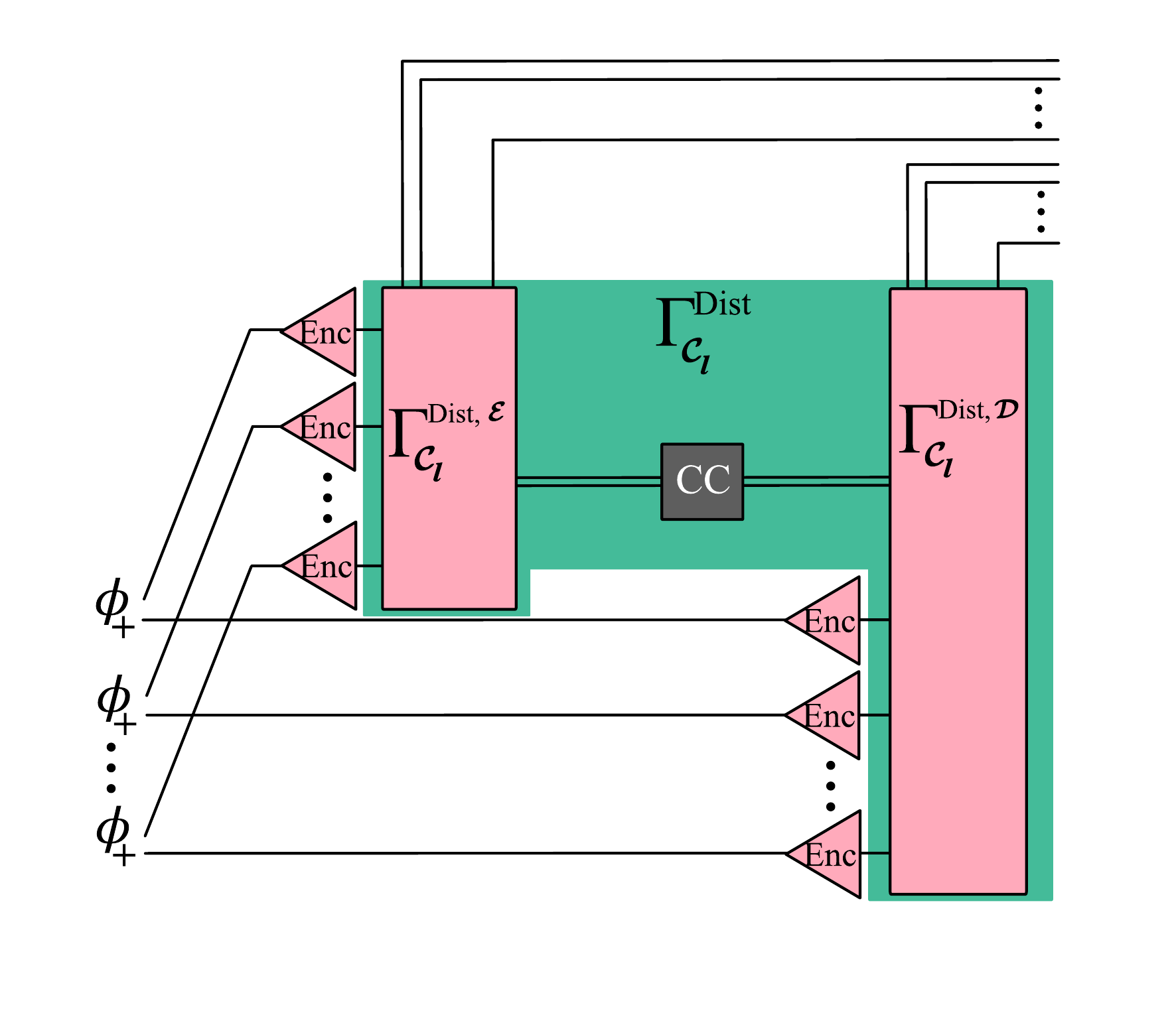}
\caption{\textbf{Setup for fault-tolerant entanglement distillation.} The local operations performed in Figure~\ref{fig-ent-dist-normally} are implemented in an error correcting code $\mathcal{C}_l$, and interfaces map the $k$ logical states into effective mixed states in the code-space (cf. Theorem~\ref{thm-ft-ent-dist-normally}).}
\label{fig-ft-ent-dist}
\end{figure}

\begin{theorem}[Fault-tolerant entanglement distillation] 
\label{thm-ft-ent-dist-normally}
For each $l\in \mathbbm{N}$, let $\mathcal{C}_l$ denote the $l$-th level of the concatenated 7-qubit Steane code with threshold $p_0$.
For any $0\leq p \leq \frac{p_0}{2}$ and for all $k\in \mathbbm{N}$ large enough, there exists a circuit 
$\Gamma^{\Dist}:\mathcal{M}_2^{\otimes 2k}\rightarrow \mathcal{M}_2^{\otimes 2\beta(4cp)k}$
 using $(1-\beta(4cp))k$ bits of classical communication, and two quantum states $\sigma^{\mathcal{E}}_S \in \mathcal{M}_{d_{S_E}}$ and $\sigma^{\mathcal{D}}_S \in \mathcal{M}_{d_{S_D}}$ such that
    \begin{IEEEeqnarray*}{Cl}
            \IEEEeqnarraymulticol{2}{l}{\|(\DecI_l^*)^{\otimes 2\beta(4cp)k}\circ \big[ \Gamma^{\Dist}_{\mathcal{C}_l}  \circ \EncI_l^{\otimes 2k} \big]_{\mathcal{F}(p)} (\phi_+)^{\otimes k} } \\& -(\phi_+)^{\otimes \beta(4cp)k}  \otimes \sigma^{\mathcal{E}}_S \otimes \sigma^{\mathcal{D}}_S\|_{\Tr}   \\ \leq& p_0 \left(\frac{p}{p_0}\right)^{2^l} |\Loc(\Gamma^{\Dist})|  +  \sqrt{ \epsilon_{dist}(4cp) } + \frac{2}{k} 
        \end{IEEEeqnarray*}
with the constant $c$ from Theorem~\ref{thm-correct-interfaces}, $\epsilon_{dist}(q)$ the function from Theorem~\ref{thm-ent-dist-original}, and $\beta(q)=1-h(q)-q\log(3)$.
\end{theorem}

The proof employs techniques from \cite{AGP05} to relate the fault-affected circuit implementations sketched in Figure~\ref{fig-ft-ent-dist} to the ideal circuits via the threshold theorem and choosing a high enough concatenation level $l$. Due to Lemma~\ref{thm-effective-encoder}, the physical input states (which are maximally entangled states here) are acted upon by the effective interface.
Thereby, they are effectively transformed into the noisy mixed states of the form $(N_{enc,p,l}\otimes N_{enc,p,l}) (\phi_+)= (1-4cp) \phi_+ + 4cp \tilde{\sigma}_l $ for some quantum state $\sigma_l$, which are twirled into Bell-diagonal form $\phi_{4cp}=(1-4cp)\phi_+ + \frac{4cp}{3} (\phi_- +\psi_+ + \psi_-)$ by the first step of the distillation protocol. For states of this form, the results from Theorem~\ref{thm-ent-dist-original} apply. The detailed proof is given in Appendix~\ref{appendix-entdistproof}.


While the apparatus described in Theorem~\ref{thm-ft-ent-dist-normally} performs fault-tolerant encoding and decoding, it still requires one-way classical communication between two parties, which is not allowed in the communication setup we investigate in the next section. For the purposes of fault-tolerant entanglement-assisted capacity, we therefore combine fault-tolerant implementation of these circuits with fault-tolerant classical communication via the channel $T$ to distill perfect maximal entanglement in the code space.
In this process, a fraction of the available channel copies is used to transmit classical communication. The protocol for this is essentially the same as the protocol from Theorem~\ref{thm-ft-ent-dist-normally}, where the classical communication between sender and receiver is not modelled by transmission over copies of the channel $\id_{cl}$, but instead transmitted by using the coding scheme from \cite[Theorem~V.8]{CMH20} as a subroutine on $\frac{h(4cp)+4cp\log(3)}{C_{\mathcal{F}(p)}(T)}$ copies of the channel $T$. For completeness, this process is sketched in Figure~\ref{fig-ft-ent-dist-via-t}.

\begin{figure}[htpb]
    \centering
    \includegraphics[width=0.5\textwidth]{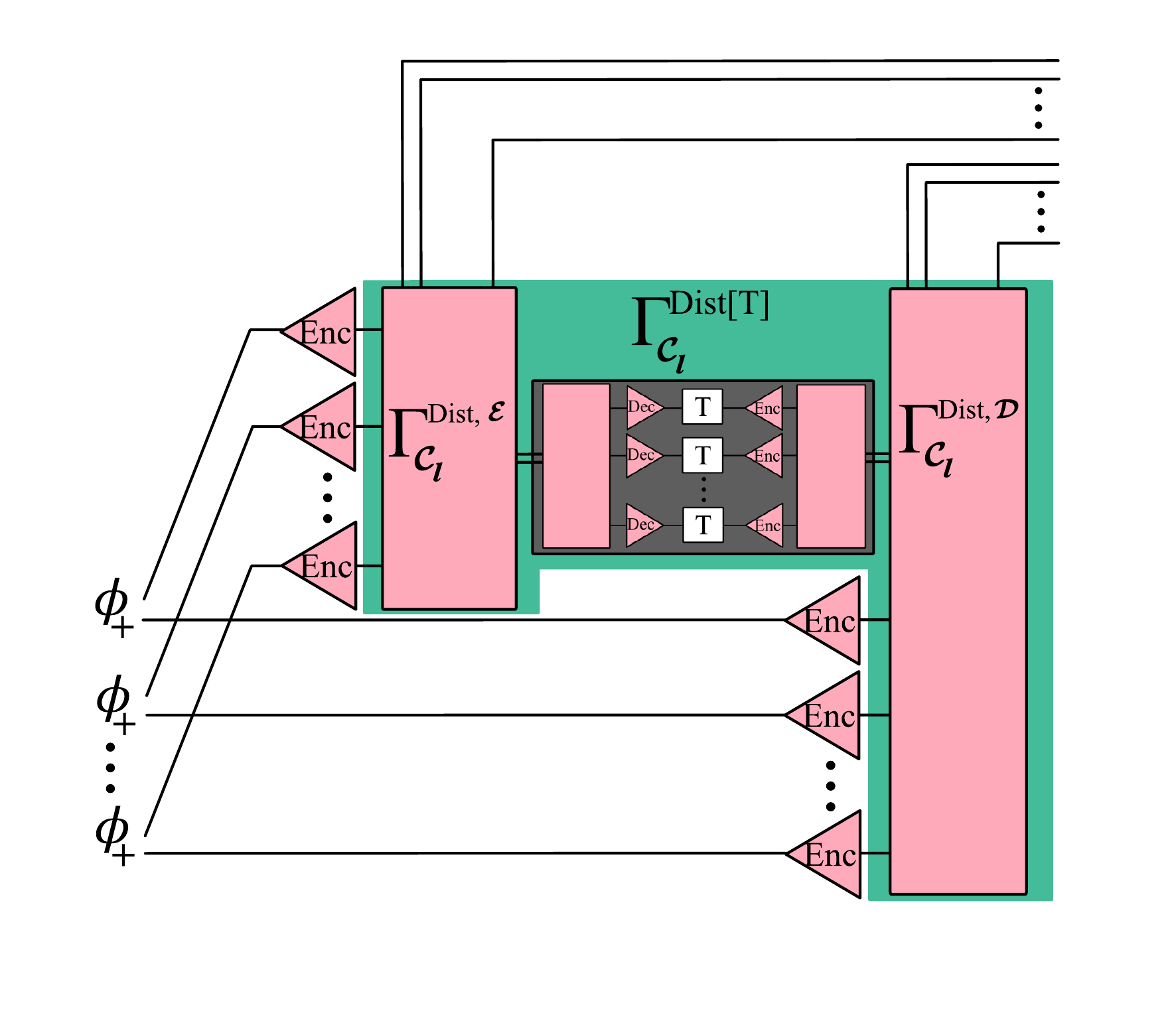}
    \caption{\textbf{Setup for fault-tolerant entanglement distillation with communication via a channel $T$.} The local operations performed in Figure~\ref{fig-ent-dist-normally} are implemented in an error correcting code $\mathcal{C}_l$, and interfaces map the $k$ logical states into effective mixed states in the code-space. The classical communication is performed as a subroutine using the coding scheme of \cite{CMH20}. Since the transmitted information is only classical, the syndrome states of these subroutines will not be correlated, allowing us to treat them as separate. 
    }
    \label{fig-ft-ent-dist-via-t}
\end{figure}

\section{A coding theorem for fault-tolerant entanglement-assisted communication} \label{sec-coding-thm}

\noindent By performing fault-tolerant entanglement distillation as a subroutine and using a subset of the channels to convert entanglement, we can thus obtain the resource entanglement for entanglement-assisted communication in the code space.
Then, the remainder of channel copies and the recovered pure state entanglement can be used for information transfer via a coding scheme for entanglement-assisted communication, contributing to the channel capacity. This subroutine can be analyzed separately from the distillation part, as sketched in Figure~\ref{fig-allparts}. If the combined coding scheme is fault-tolerant, then the information transfer is fault-tolerant.

As outlined in Section~\ref{sec-ft-cap-setup}, the coding scheme we will use after the distillation in the fault-tolerant setting is based on the scheme used for an effective noisier channel in the faultless setting which takes a correlated syndrome state as part of its input. More precisely, we use the coding scheme for the effective channel to prepare the codeword states in the logical subspace. Then, our results on entanglement-assisted capacity under AVP apply in order to obtain bounds on achievable rates in the presence of correlated syndrome states.


Note that we have the following upper bound for fault-tolerantly achievable rates:

\begin{theorem}\label{thm-upper-bound} Let $p\geq0$. For every quantum channel $T:\mathcal{M}_{d_A}\rightarrow \mathcal{M}_{d_B}$, we have 
    \[C^{ea}_{\mathcal{F}(p)}(T) \leq(1-p) C^{ea}(T).\]
\end{theorem}

\begin{IEEEproof}[Proof of Theorem~\ref{thm-upper-bound}]
   Notably, the bound $C^{ea}(T)\geq C^{ea}_{\mathcal{F}(p)}(T)$ trivially holds for any channel $T$. Let $\mathcal{F}'(p)$ be a fault-model where no gate is affected by an error except for the gates applied right after the communication channel $T$. Clearly, we have $C^{ea}_{\mathcal{F}(p)}(T)\leq C^{ea}_{\mathcal{F}'(p)}(T) =C^{ea}(\tilde{T}_p)$ where $\tilde{T}_p=D_p\circ T= (1-p)T+p \mathbbm{1}/d \Tr$ with the depolarizing channel $D_p$ and depolarizing probability $p$. Then, we have $C^{ea}(\tilde{T}_p)\leq (1-p)C^{ea}(T)+p C^{ea} (\mathbbm{1}/d \Tr )=(1-p)C^{ea}(T)$ because of the convexity of mutual information and because the entanglement-assisted capacity of the channel $\mathbbm{1}/d \Tr$ (a replacer channel which always outputs the maximally mixed state) is zero. 
\end{IEEEproof}


Here, we derive a lower bound in the form of a threshold theorem for fault-tolerant entanglement-assisted communication for any quantum channel $T$, where the fault-tolerant entanglement-assisted capacity approaches the usual, faultless case for vanishing gate error probability:

\begin{theorem}[Threshold theorem for fault-tolerant entanglement-assisted communication] \label{thm-main-result-1} 
For every quantum channel $T:\mathcal{M}_{d_A}\rightarrow \mathcal{M}_{d_B}$, and any $\eta>0$, there exists a threshold $p(\eta,T)> 0$ such that, for any $0\leq p\leq p(\eta,T)$, we have
\[C^{ea}_{\mathcal{F}(p)}(T) \geq C^{ea}(T)-\eta\]
\end{theorem}

We note that the theorem is threshold-like in that the traditional capacity can be arbitrarily approximated when the gate noise is below a threshold that is derived from a fault-tolerant threshold. In the according formulation of this capacity approximation, however, the threshold $p(\eta,T)$ depends not only on the channel but also on the required approximation $\eta$.

\begin{corollary} \label{thm-main-result-2}
Let $p\geq0$. Then, for every quantum channel $T:\mathcal{M}_{d_A}\rightarrow \mathcal{M}_{d_B}$, we have 
\[\lim_{p\rightarrow 0} C^{ea}_{\mathcal{F}(p)}(T) = C^{ea}(T) \]
\end{corollary}

This is a consequence of the following result (noting that $C(T)=0$ implies $C^{ea}(T)=0$):

\begin{theorem}[Lower bound on the fault-tolerant entanglement-assisted capacity] \label{thm-final-coding-thm}
For $0\leq p \leq 1$, let $\mathcal{F}(p)$ denote the
i.i.d. Pauli noise model and let $0\leq p_0 \leq 1$ denote the threshold of the concatenated 7-qubit Steane code.
For any quantum channel $T:\mathcal{M}_2^{\otimes j_1} \rightarrow \mathcal{M}_2^{\otimes j_2}$ with classical capacity $C(T)>0$ and for any $0\leq p \leq \min\{p_0/2,1/(2c(j_1+j_2)\}$, we have a fault-tolerant entanglement-assisted capacity with 
\[C^{ea}_{\mathcal{F}(p)}(T) \geq  C^{ea}(T)- 4f_1(p)\frac{C^{ea}(T)}{C(T)} -f_2(p) \]
where 
\begin{equation*}\label{f_1}
    f_1(p)= \frac{(h(4cp)+4cp\log(3))j_2 }{1-h(4cp)-4cp\log(3)},
\end{equation*}
and

\begin{IEEEeqnarray*}{lCl}
f_2(p)
& = &2\sqrt{2{j_2p}} \big(2^{j_1+j_2} (j_1+j_2) +1\big) \big|2\log( \frac{2(j_1+j_2)cp }{2^{j_1j_2}})\big|  \\&&  +2h\big(\sqrt{2{j_2p}} 2^{j_1+j_2} |2\log( \frac{2(j_1+j_2)cp }{2^{j_1j_2}})| \big) \\&&
 +(1+4(j_1+j_2)cp)h\big(\frac{4(j_1+j_2)cp}{1+4(j_1+j_2)cp}\big) \\&&   +10 (j_1+j_2)cpj_2 ,
\end{IEEEeqnarray*}
and with $c$ being the constant from Theorem~\ref{thm-correct-interfaces}.
\end{theorem}

\begin{IEEEproof}[Proof of Theorem~\ref{thm-final-coding-thm}]
In this proof, we construct a fault-tolerant coding scheme for entanglement-assisted communication affected by the Pauli i.i.d. noise model $\mathcal{F}(p)$ by proposing an encoder circuit $\Enc$ and a decoder circuit $\Dec$ which are implemented in the concatenated 7-qubit Steane code ${\mathcal{C}_l}$ with threshold $p_0$ and for some level $l$. With a rate of entanglement assistance $R_{ea}$, we will obtain a bound on rates $R$ that fulfill
\begin{IEEEeqnarray}{ll}
    \IEEEeqnarraymulticol{2}{l}{\|\id_{cl}^{\otimes nR}-   \big[ \Gamma_{\mathcal{C}_l}^{\mathcal{D}} \circ  ((\EncI_l \circ T \circ \DecI_l)^{\otimes n}\otimes \id_2^{*\otimes nR_{ea}}) \circ \Compactcdots} \IEEEnonumber\\ \Compactcdots \circ (\Gamma_{\mathcal{C}_l}^{\mathcal{E}} \otimes \id_2^{*\otimes nR_{ea}} ) \circ \Compactcdots &  \label{this-has-to-go-to-zero} \\ \Compactcdots \circ  (\id_{cl}^{\otimes nR} \otimes (\EncI^{\otimes 2}(\phi_+))^{\otimes nR_{ea}}) \big]_{\mathcal{F}(p)}\|_{1 \rightarrow 1} & \overset{n\rightarrow \infty}{\rightarrow}  0 ,\IEEEnonumber
\end{IEEEeqnarray}
showing which rates $R$ are achievable.

Our proof will progress according the following strategy:

\begin{enumerate}
\item Construct the coding scheme out of the relevant subcircuits for distillation and coding under arbitrarily varying perturbations, as illustrated in Figure~\ref{fig-allparts}.
\item Choose the concatenation level $l$ corresponding to the number of locations in the entire coding scheme, including all subcircuits, in Eq.~\eqref{eq-level-choice}.
\item Bound expression Eq.~\eqref{this-has-to-go-to-zero} in terms of the effective channel and the distilled state using Theorem~\ref{thm-eff-channel} and Theorem~\ref{thm-ft-ent-dist-normally}.
\item Invoke our results on entanglement-assisted capacity under arbitrarily varying perturbations of strength $2(j_1+j_2)cp$ from Theorem~\ref{thm-avp-coding-thm} to obtain a bound on the achievable rates.
\end{enumerate}

The coding scheme for fault-tolerant communication is based on the coding scheme for communication  at a rate $R_{AVP}$ under arbitrarily varying perturbation.
For each $n\in\mathbbm{N}$, using the Solovay-Kitaev theorem \cite{NC00}, we may choose specific quantum circuits $\Gamma^{\AVP,\mathcal{E}}$ and $\Gamma^{\AVP,\mathcal{D}}$ implementing the encoder $\mathcal{E}$ and decoder $\mathcal{D}$ used for communication under AVP, such that
\begin{equation}\label{eq-Enc}
\|\Gamma^{\AVP,\mathcal{E}} - \mathcal{E}\|_{1 \rightarrow 1} \leq \frac{1}{ n} \end{equation}
\begin{equation}\label{eq-Dec} \|\Gamma^{\AVP,\mathcal{D}}- \mathcal{D}\|_{1 \rightarrow 1} \leq \frac{1}{ n}\end{equation}

In addition, let $\Gamma^{\Dist[T]}$ be the circuit performing entanglement distillation, which is based on the circuit of Corollary~\ref{thm-ent-dist-and-dil}. This circuit is based on the circuit from Theorem~\ref{thm-ft-ent-dist-normally} with classical communication via a subset of the copies of the channel $T$, and an additional step of entanglement dilution to distill the bipartite pure entangled state $\varphi$ using the protocol from \cite{LP99}, which requires additional one-way classical communication at an asymptotically negligible rate so long as the state is pure (see also \cite[Footnote~1]{BSST02}). Thereby, similar to the sketch in Figure~\ref{fig-ft-ent-dist-via-t}, the fault-tolerant implementation of this distillation circuit distills $\varphi$ in the code space, whereby it is made available for the fault-tolerantly implemented communication setup. This is formalized in Appendix~\ref{appendix-entdistproof} and Corollary~\ref{thm-ent-dist-and-dil}.

Let $\Gamma^{\AVP,\mathcal{E}}_{\mathcal{C}_l} $, $\Gamma^{\AVP,\mathcal{D}}_{\mathcal{C}_l} $ and $\Delta_{\mathcal{C}_l}^{[T]}$ denote the implementations of these circuits in the 7-qubit Steane code with concatenation level $l$. 
The circuits $ \Gamma_{\mathcal{C}_l}^{\mathcal{E}}$ and $ \Gamma_{\mathcal{C}_l}^{\mathcal{D}}$ which implement our proposed fault-tolerant coding scheme are then constructed by the local parts of fault-tolerant entanglement distillation, followed by the fault-tolerant implementation of the coding scheme for arbitrarily varying perturbations, as sketched in Figure~\ref{fig-allparts}.

In total, the maximally entangled resource states are transformed by the noisy interface into effective noisy states in the code space. Thereafter, entanglement distillation is performed to restore pure state entanglement in the code space, using up a subset of the copies of $T$ for classical communication. Subsequently, the remaining copies of $T$ are used for entanglement-assisted communication.

By our construction, we obtain the following expression for the coding error Eq.~\eqref{this-has-to-go-to-zero} in terms of the subroutines sketched in Figure~\ref{fig-allparts}, corresponding to the entanglement distillation (which uses $\alpha_p n$ channel copies and produces $\beta_p$ entangled states) and the coding scheme under AVP via the remaining $(1-\alpha_p)n$ channel copies:

\begin{IEEEeqnarray}{ll}
\IEEEeqnarraymulticol{2}{l}{ \| \id_{cl}^{\otimes nR } -  \big[ \Gamma_{\mathcal{C}_l}^{\mathcal{D}} \circ \Compactcdots}   \label{eq-the-very-long-calculation1}\\& \Compactcdots \circ  \Big( \big( (\EncI_l \circ T \circ \DecI_l)^{\otimes n} \circ \Gamma_{\mathcal{C}_l}^{\mathcal{E}} \big)\otimes \id_2^{*\otimes nR_{ea}} \Big)\circ \Compactcdots \IEEEnonumber\\& \Compactcdots \circ  \big(\id_{cl}^{\otimes nR} \otimes ( \EncI_l^{\otimes 2}(\phi_+))^{\otimes nR_{ea}} \big) \big]_{\mathcal{F}(p)}\|_{1 \rightarrow 1} 
\IEEEnonumber\\
   =& \| \id_{cl}^{\otimes nR }-     \big[ \Gamma^{\AVP,\mathcal{D}}_{\mathcal{C}_l}  \circ  \Big( \big( (\EncI_l \circ T \circ \DecI_l)^{\otimes (1-\alpha_p)n} \circ \Compactcdots \IEEEnonumber\\& \Compactcdots \circ \Gamma^{\AVP,\mathcal{E}}_{\mathcal{C}_l} \big)\otimes \id_2^{*\otimes \beta_{p}n} \Big) \circ \Compactcdots \IEEEnonumber\\& \Compactcdots \circ \big(\id_{cl}^{\otimes nR} \otimes \Gamma^{\Dist[T]}_{\mathcal{C}_l}  ( \EncI_l^{\otimes 2}(\phi_+))^{\otimes nR_{ea}} \big) \big]_{\mathcal{F}(p)}\|_{1 \rightarrow 1} .\IEEEnonumber
\end{IEEEeqnarray} 
Here, we use \[\alpha_{p}=\frac{h(4cp)+4cp\log(3)}{C_{\mathcal{F}(p)}(T)} R_{ea}\] in order to distribute the $n$ total channel copies such that $\alpha_p n$ copies of $T$ are used in the distillation subroutine, and $(1-\alpha_p)n$ copies of $T$ are used for the communication subroutine. We furthermore use \[\beta_{p}=\frac{1-h(4cp)-4cp\log(3)}{H(A)_{\varphi}}R_{ea} \] 
for the number of maximally entangled states that the distillation produces, in reference to the notation in Theorem~\ref{thm-ft-ent-dist-normally} and \ref{thm-ent-dist-and-dil}.

\begin{figure}[htbp] 
  \centering
       \includegraphics[width=0.5\textwidth]{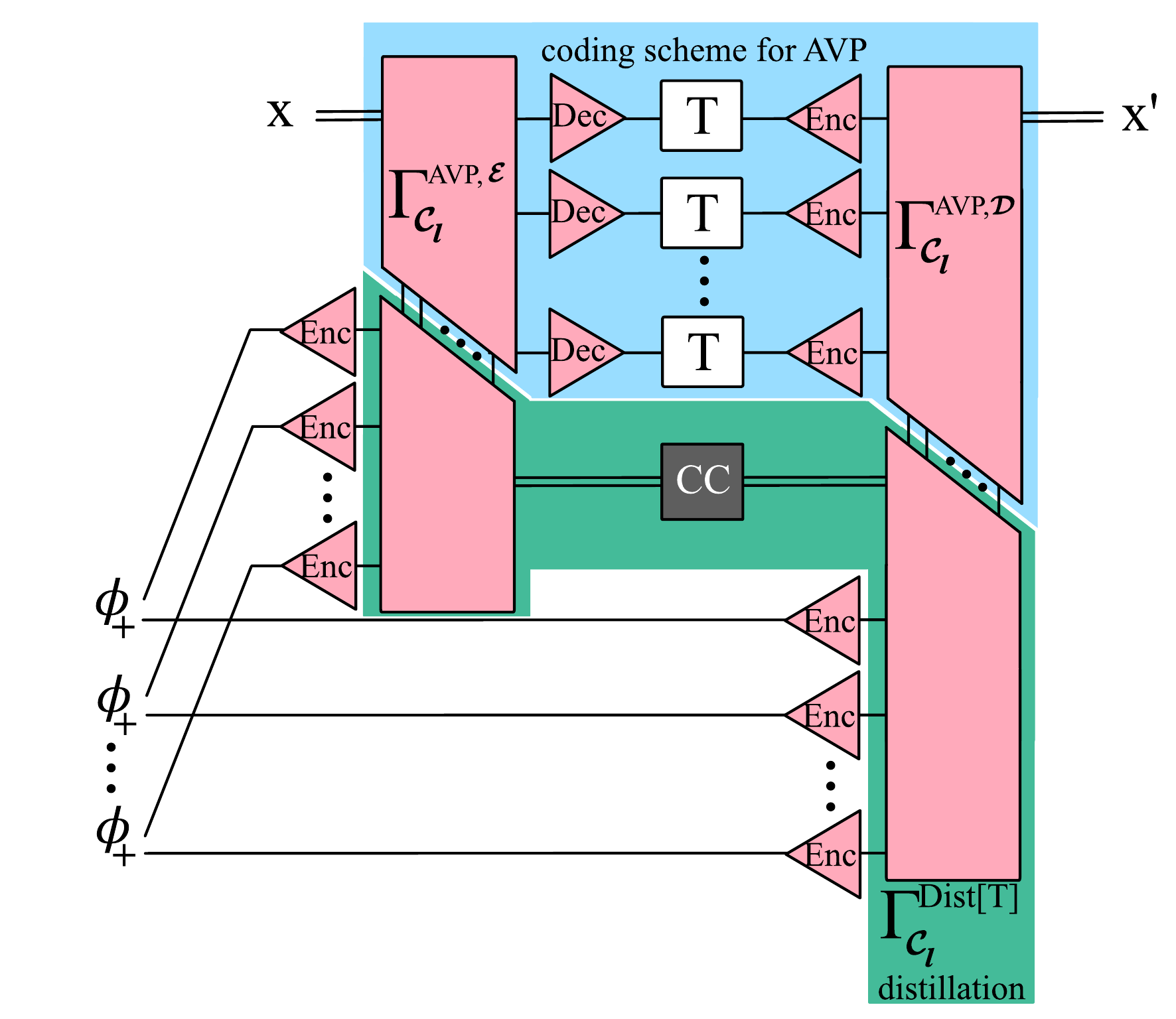}
\caption{\textbf{An illustration of the building blocks in our analysis.} The circuits for encoding and decoding in our fault-tolerant entanglement-assisted communication setup, as illustrated in Figure~\ref{fig-ft-ea-cap-setup}, are constructed out of 3 individual subroutines. Firstly, one subroutine (indicated in green) consists of the local operations performed during entanglement distillation. This part of the analysis will be based on our results in Theorem~\ref{thm-ft-ent-dist-normally}. The classical communication between the local distillation parts makes up another subroutine (indicated by a black box), and is performed using the coding scheme proposed in \cite{CMH20}, as sketched in Figure~\ref{fig-ft-ent-dist-via-t}.
The final subroutine (indicated in blue) is given by the fault-tolerant implementation of the coding scheme for 
entanglement-assisted communication under AVP, as described in Theorem~\ref{thm-avp-coding-thm}. Note that the distillation subroutine and the classical communication subroutine are only connected by classical information, while the coding scheme and the distillation subroutine are connected by quantum states in the code. In total, these subroutines and their analysis are combined in Eq.~\eqref{eq-the-very-long-calculation3}.}
\label{fig-allparts}
\end{figure}
For any sequence of circuits, we can choose the Steane code concatenation level $l=l_{ n}$ high enough such that the implementations above fulfill
\begin{align}\label{eq-level-choice}\begin{split}
    \left(\frac{p}{p_0}\right)^{2^{l_{n}-1}} \big(|\Loc(\Gamma^{\AVP,\mathcal{E}})|+|\Loc(\Gamma^{\AVP,\mathcal{D}})|&\\+j_1  n + |\Loc(\Gamma^{\Dist})|\big) & \leq \frac{1}{n}. 
\end{split}
\end{align}
We emphasize that this choice of $l$ does not impose a restriction on the circuits we consider; rather, the sizes of the circuits for a given channel limit how low $l$ can be chosen to be.

Using Theorem~\ref{thm-eff-channel}, Eq.~\eqref{eq-the-very-long-calculation1} can be bounded in terms of the effective channel $T_{p,N_l}=(1-2(j_1+j_2) cp) (T\otimes \Tr_S) + 2(j_1+j_2)cpN_l$ for some quantum channel $N_l$, and thereby connected to our results on capacity under AVP for perturbation probability $2(j_1+j_2) cp$:
\begin{IEEEeqnarray*}{Cl}
 \IEEEeqnarraymulticol{2}{l}{
    \| \id_{cl}^{\otimes nR }-     \big[ \Gamma^{\AVP,\mathcal{D}}_{\mathcal{C}_l}  
\circ  \Big( \big( (\EncI_l \circ T \circ \DecI_l)^{\otimes (1-\alpha_p)n} \circ\Compactcdots }\IEEEnonumber\\&\Compactcdots \circ  \Gamma^{\AVP,\mathcal{E}}_{\mathcal{C}_l} \big)\otimes \id_2^{*\otimes \beta_{p}n} \Big)  
\circ \Compactcdots \IEEEnonumber\\& \Compactcdots \circ \big(\id_{cl}^{\otimes nR} \otimes \Gamma^{\Dist[T]}_{\mathcal{C}_l}  ( \EncI_l^{\otimes 2}(\phi_+))^{\otimes nR_{ea}} \big) \big]_{\mathcal{F}(p)}\|_{1 \rightarrow 1} 
\IEEEnonumber\\
 \leq&\| \id_{cl}^{\otimes nR } -   (\Gamma^{\AVP,\mathcal{D}}\otimes \Tr_S)  \circ  \Big( T_{p,N_l}^{\otimes (1-\alpha_p)n} \circ\Compactcdots \\& \Compactcdots \circ  (\Gamma^{\AVP,\mathcal{E}} \otimes S_S )\big)\otimes \id_2^{*\otimes \beta_{p}n} \Big) \circ \Compactcdots \IEEEnonumber\\& \Compactcdots \circ  \big(\id_{cl}^{\otimes nR} \otimes (\DecI_l^*)^{\otimes 2\beta_{p}n} \circ\Compactcdots \IEEEnonumber\\& \Compactcdots \circ  \big[ \Gamma^{\Dist[T]}_{\mathcal{C}_l} ( \EncI_l^{\otimes 2}(\phi_+))^{\otimes nR_{ea}}\big]_{\mathcal{F}(p)}  \big) \|_{1 \rightarrow 1}  \IEEEnonumber\\ & +
2p_0 \left(\frac{p}{p_0}\right)^{2^{l_{n}-1}} (|\Loc(\Gamma^{\AVP,\mathcal{E}})|+|\Loc(\Gamma^{\AVP,\mathcal{D}})|+j_1  n) \IEEEnonumber\label{eq-the-very-long-calculation2}
\end{IEEEeqnarray*}
Note that the distillation part of the circuit $\Gamma^{\Dist[T]}_{\mathcal{C}_l}$ ends in an error correction gadget in the lines where there is quantum output. This error correction gadget features in the effective channel theorem in Lemma~\ref{thm-eff-channel}, where it is used for the transformation rules linking the fault-affected circuit with a noiseless version. While this error correction gadget is important for the transformation rules and fault-analysis, it remains unchanged by the process, which is why it appears in both expressions in Lemma~\ref{thm-eff-channel}. Then, it can be recombined with the rest of the distillation circuit in order to recover  $\Gamma^{\Dist[T]}_{\mathcal{C}_l}$, which is why the error correction gadget does not explictly appear in the inequality above. Note also that $\Tr_S(\sigma_S)=1$ for any syndrome state.

Then, Theorem~\ref{thm-ft-ent-dist-normally} is employed to perform entanglement distillation of the bipartite entangled states $\varphi$ in the code space, leading to the following transformation:
\begin{IEEEeqnarray}{Cl}
 \IEEEeqnarraymulticol{2}{l}{ \| \id_{cl}^{\otimes nR } -   (\Gamma^{\AVP,\mathcal{D}}\otimes \Tr_S)  \circ \Compactcdots}\IEEEnonumber\\&\Compactcdots \circ   \Big( T_{p,N_l}^{\otimes (1-\alpha_p)n} \circ (\Gamma^{\AVP,\mathcal{E}} \otimes S_S )\big)\otimes \id_2^{*\otimes \beta_{p}n} \Big) \circ \Compactcdots \IEEEnonumber\\&\Compactcdots \circ  \big(\id_{cl}^{\otimes nR} \otimes (\DecI_l^*)^{\otimes 2\beta_{p}n} \circ \Compactcdots \IEEEnonumber\\&\Compactcdots \circ   \big[ \Gamma^{\Dist[T]}_{\mathcal{C}_l} ( \EncI_l^{\otimes 2}(\phi_+))^{\otimes nR_{ea}}\big]_{\mathcal{F}(p)}  \big) \|_{1 \rightarrow 1}  \IEEEnonumber\\&+
2p_0 \left(\frac{p}{p_0}\right)^{2^{l_{n}-1}} (|\Loc(\Gamma^{\AVP,\mathcal{E}})|+|\Loc(\Gamma^{\AVP,\mathcal{D}})|+j_1  n) \IEEEnonumber\\
 \leq &\| \id_{cl}^{\otimes nR } -   (\Gamma^{\AVP,\mathcal{D}}\otimes \Tr_S) \circ  \Big( T_{p,N_l}^{\otimes (1-\alpha_p)n} \circ \Compactcdots \IEEEnonumber\\& \Compactcdots \circ   (\Gamma^{\AVP,\mathcal{E}} \otimes S_S )\big)\otimes \id_2^{*\otimes\beta_{p}n } \Big)\circ \Compactcdots \IEEEnonumber\\&\Compactcdots \circ (\id_{cl}^{\otimes nR} \otimes \varphi^{\otimes \beta_{p}n} \otimes \sigma_S)  \|_{1 \rightarrow 1}  \IEEEnonumber\\& +
2p_0 (\frac{p}{p_0})^{2^{l_{n}-1}} (|\Loc(\Gamma^{\AVP,\mathcal{E}})|+|\Loc(\Gamma^{\AVP,\mathcal{D}})|+j_1  n)  \IEEEnonumber\\&+2p_0 (\frac{p}{p_0})^{2^{l_{n}-1}} |\Loc(\Gamma^{\Dist})|)+ \frac{2}{nR_{ea}}+\sqrt{\epsilon'_{dist}(4cp)} 
\IEEEnonumber\\
 \leq &\| \id_{cl}^{\otimes nR } -    \Gamma^{\AVP,\mathcal{D}}  \circ  \Big( T_{p,N_l}^{\otimes (1-\alpha_p)n} \circ \Compactcdots \IEEEnonumber\\&\Compactcdots \circ   (\Gamma^{\AVP,\mathcal{E}} \otimes S_S(\sigma_S) )\big) \otimes \id_2^{*\otimes \beta_{p}n} \Big) \circ \Compactcdots \IEEEnonumber\\&\Compactcdots \circ  (\id_{cl}^{\otimes nR} \otimes \varphi^{\otimes \beta_{p}n} ) \|_{1 \rightarrow 1}  \IEEEnonumber\\& +
\frac{1}{n}+\frac{2}{nR_{ea}}+\sqrt{\epsilon'_{dist}(4cp)}.     \label{eq-the-very-long-calculation3}
\end{IEEEeqnarray}
For the first inequality, we used Corollary~\ref{thm-ent-dist-and-dil} to perform entanglement distillation on the $k=nR_{ea}$ entangled states $\big[ \EncI_l^{\otimes 2}\big]_{\mathcal{F}(p)} (\phi_+)
$ that have been affected by the noisy interface, obtaining the bipartite pure entangled state $\varphi$ in the code space. In total, this distillation process uses up $\alpha_{p}n$ copies of $T$ in the process to perform classical communication and produces $\beta_{p}n$ copies of $\varphi$ in the code space, where the explicit expressions for $\alpha_{p}$ and $\beta_p$ can be found to originate from Corollary~\ref{thm-ent-dist-and-dil}.

The second inequality is a consequence of our choice of concatenation level in Eq.~\eqref{eq-level-choice}.

We now use Eq.~\eqref{eq-Enc} and Eq.~\eqref{eq-Dec} to relate the circuits for the coding scheme in Eq.~\eqref{eq-the-very-long-calculation3} to the ideal operations:

\begin{IEEEeqnarray*}{Cl}
 \IEEEeqnarraymulticol{2}{l}{ \| \id_{cl}^{\otimes nR } -    \Gamma^{\AVP,\mathcal{D}}  \circ  \Big( T_{p,N_l}^{\otimes (1-\alpha_p)n} \circ\Compactcdots  }  \\& \Compactcdots \circ  (\Gamma^{\AVP,\mathcal{E}} \otimes S_S(\sigma_S) )\big)\otimes \id_2^{*\otimes \beta_{p}n} \Big)  \circ\Compactcdots \\& \Compactcdots \circ (\id_{cl}^{\otimes nR} \otimes \varphi^{\otimes \beta_{p}n} ) \|_{1 \rightarrow 1} \\& +
\frac{1}{n}+\frac{2}{nR_{ea}}+\sqrt{\epsilon'_{dist}(4cp)}\\
\leq & \| \id_{cl}^{\otimes nR } -  \mathcal{D}  \circ \Big( \big( T_{p,N_l}^{\otimes (1-\alpha_p)n}\circ\Compactcdots \\& \Compactcdots \circ  (\mathcal{E} \otimes S_S(\sigma_S) )\big)\otimes \id_2^{*\otimes \beta_{p}n} \Big)  \circ (\id_{cl}^{\otimes nR} \otimes \varphi^{\otimes \beta_{p}n} )  \|_{1 \rightarrow 1}   \\&+\frac{2}{nR_{ea}}+\sqrt{\epsilon'_{dist}(4cp)}+\frac{3}{n}  
\end{IEEEeqnarray*}
Finally, we note that $\epsilon'_{dist}(4cp)$ goes to zero as $n\rightarrow \infty$, and so do $\frac{3}{ n}$ and $\frac{2}{nR_{ea}}$. The remaining summand goes to zero for all achievable rates $R_{AVP}$ of entanglement-assisted communication under AVP with perturbation probability $2(j_1+j_2)cp$ and quantum channel $N_l$. For any entanglement-assistance at rate $R_{ea}'\geq 1$, these achievable rates are described by the expression in Eq.~\eqref{eq-r-avp}. Since a fraction of the copies of $T$ are used in the distillation, the communication rate is thus reduced. In total, we find the following bound on the achievable rate of fault-tolerant entanglement-assisted communication:
\[R < (1-\alpha_p) R_{AVP}\]
The bound on $R_{ea}'$ also automatically implies that $R_{ea}\geq \frac{1}{\beta_{p}}= \frac{H(A)_{\varphi}}{1-h(4cp)-4cp\log(3)}$. In order to simplify notation in the main theorem, and since additional entanglement does not increase capacity, we will henceforth set this rate to  \[R_{ea} := \frac{\log(2)}{1-h(4cp)-4cp\log(3)}\]
This leads to a fault-tolerant entanglement-assisted capacity of
\begin{IEEEeqnarray*}{lCl}
C^{ea}_{\mathcal{F}(p)}(T) &\geq &(1-\alpha_p) C_{AVP}((j_1+j_2)cp,T)\\& \geq & C^{ea}(T)- f_1(p)\frac{C^{ea}(T)}{C_{\mathcal{F}(p)}(T)} -f_2(p) \end{IEEEeqnarray*}
where
\begin{equation*}
    f_1(p)= \frac{(h(4cp)+4cp\log(3))\log(2) }{1-h(4cp)-4cp\log(3)}
\end{equation*}
and
\begin{IEEEeqnarray*}{lCl}
f_2(p)
& = &2\sqrt{2{j_2p}} \big(2^{j_1+j_2} (j_1+j_2) +1\big) \big|2\log( \frac{2(j_1+j_2)cp }{2^{j_1j_2}})\big|  \\&&  +2h\big(\sqrt{2{j_2p}} 2^{j_1+j_2} |2\log( \frac{2(j_1+j_2)cp }{2^{j_1j_2}})| \big) \\& &
 +(1+4(j_1+j_2)cp)h\big(\frac{4(j_1+j_2)cp}{1+4(j_1+j_2)cp}\big) \\ &&   +10 (j_1+j_2)cpj_2 ,
\end{IEEEeqnarray*}
Using \cite[Theorem~V.8]{CMH20}, for any channel $T$ with classical capacity $C(T)>0$, we find an explicit function $0\leq f(p)\leq C(T)$ such that we have
\begin{IEEEeqnarray*}{lCl}
    C^{ea}_{\mathcal{F}(p)}(T)& \geq&  C^{ea}(T)- f_1(p)\frac{C^{ea}(T)}{C(T)-f(p)} -f_2(p)\\ & 
     \geq& C^{ea}(T)- f_1(p)\frac{C^{ea}(T)}{C(T)} \frac{1}{1-\frac{f(p)}{C(T)}} -f_2(p)\\&
    \geq& C^{ea}(T)- 2f_1(p)\frac{C^{ea}(T)}{C(T)}\Big(1+\frac{f(p)}{C(T)}\Big) -f_2(p)\\&
    \geq& C^{ea}(T)- 4f_1(p)\frac{C^{ea}(T)}{C(T)} -f_2(p)
\end{IEEEeqnarray*}
\end{IEEEproof}

Theorem~\ref{thm-main-result-1} is obtained as a direct consequence of this result:

\begin{IEEEproof}[Proof of Theorem~\ref{thm-main-result-1}]
For a given quantum channel $T$, we have 
$C^{ea}_{\mathcal{F}(p)}(T) \geq  C^{ea}(T)- 4f_1(p)\frac{C^{ea}(T)}{C(T)} -f_2(p) $ with the functions from Theorem~\ref{thm-final-coding-thm}.

Then, for any $\epsilon>0$, we can find a $p_0(T,\epsilon)$ such that $4f_1(p)\frac{C^{ea}(T)}{C(T)} +f_2(p)\leq\epsilon $ for all $0\leq p\leq p_0(T,\epsilon)$.
\end{IEEEproof}

It should be noted that the bound in Theorem~\ref{thm-final-coding-thm} is dependent on the individual channel $T$ and not only on its dimension, which would lead to a uniform convergence statement. Uniformity would follow if the quotient of classical and entanglement-assisted capacity were bounded for a given dimension, as has been conjectured in \cite{BSST02}.

\section{Conclusion and open problems}

\noindent The usual notion of capacity of a channel considers a perfect encoding of information into the channel, transfer via the (noisy) channel, and subsequent decoding. In real-world devices, this process of encoding and decoding the information cannot be assumed to be free of faults, which suggests the necessity of a modified notion of capacity.
Here, we show that entanglement-assisted transfer of information is possible at almost the same rates for fault-affected devices as long as the probability for gate error is below a threshold. 

Coding theorems can be understood as a conversion between resources, where quantum channels, entangled states and classical channels are used to simulate one another. Based on the notation from \cite{DHW08}, we say $\alpha \geq_{FT(p)} \beta$ if there exists a fault-tolerant transformation from a resource $\alpha$ to a resource $\beta$ at gate error $p$ with asymptotically vanishing overall error. 
Then, our coding theorem in Theorem~\ref{thm-final-coding-thm} for fault-tolerant entanglement-assisted communication via a quantum channel $T:\mathcal{M}_{d_{A}} \rightarrow \mathcal{M}_{d_{B}}$ corresponds to the following resource inequality: for any pure state $\varphi$ on $\mathcal{M}_{d_A}\otimes \mathcal{M}_{d_{A'}}$, we have
\begin{IEEEeqnarray*}{lC} \IEEEeqnarraymulticol{2}{l}{ \langle T \rangle + \big(H(A)_{(T\otimes \id_{A'})(\varphi)}+ \mathcal{O}(p) \big) [qq] }\\ &\geq_{FT(p)} \big(I(A':B)_{(T\otimes \id_{A'})(\varphi)}+\mathcal{O}(p \log(p)) \big) [c\rightarrow c]  \end{IEEEeqnarray*}
which specifies the asymptotic resource trade-off for fault-tolerant entanglement-assisted communication.
For vanishing gate error $p$, this reduces to the standard resource inequality from \cite[Eq.~54]{DHW08}.

The presented results can be understood as a further development of the toolbox of quantum communication with noisy encoding and decoding devices. Even though we chose to present our work in the frame of an explicitly chosen setup for fault-tolerant computation (i.e. 7-qubit Steane code and i.i.d. Pauli noise), the buildup is modular in nature, allowing for the adaptation to other fault-tolerant scenarios. 

We envision that our treatment of entanglement distillation with noisy devices will find application in other quantum communication contexts (e.g. connecting quantum computers, quantum repeaters, multiparty quantum communication and quantum cryptography).

As it is not covered by the presented techniques, we leave the study of fault-tolerant communication via infinite-dimensional quantum channels for future work. 



\appendices
\counterwithin*{equation}{section}
\renewcommand\theequation{\thesection\arabic{equation}}

\section{Coding error for entanglement-assisted communication} \label{appendix-codingerror}

\begin{theorem}
For any quantum channel $T_p: \mathcal{M}_{d_A} \rightarrow \mathcal{M}_{d_B}$, and any pure bipartite quantum state $\varphi\in \mathcal{M}_{d_A}\otimes \mathcal{M}_{d_{A'}}$, there exists an encoder $\mathcal{E}:\mathbbm{C}^{2^m}\otimes\mathcal{M}_2^{ \otimes \lfloor nR_{ea} \rfloor }\rightarrow\mathcal{M}_{d_A}^{\otimes n}$ and a decoder $\mathcal{D}:\mathcal{M}_{d_B}^{\otimes n}\otimes \mathcal{M}_2^{ \otimes \lfloor nR_{ea} \rfloor } \rightarrow \mathbbm{C}^{2^m}$ for $T_p$ such that
\begin{IEEEeqnarray*}{ll}
   &  \Xi( T_p^{\otimes n })\\&:= F\Big(X,\mathcal{D} \circ  \big( ( T_p^{\otimes n } \circ \mathcal{E} ) \otimes \id_{2}^{\otimes nR_{ea}} \big) (X \otimes \varphi^{\otimes nR_{ea} })\Big)  \\&\geq 1-{\epsilon_{ea}} 
\end{IEEEeqnarray*} 
for any classical message $x$ with the corresponding quantum state $X=\ketbra{x}$ of length $nR'$, where 
\begin{IEEEeqnarray*}{lCl}\epsilon_{ea}&\leq &
12 e^{-\frac{n\delta^2}{2(\log(\lambda_{\min}))^2}}\\&&  + 16 \cdot 2^{-n(I(A':B)_{(T_p\otimes \id_2)(\varphi)} -4\delta-\omega(n,\delta) -R')}\end{IEEEeqnarray*}
with the function $\eta(\delta,d_A,d_B)= (d_Ad_B\log(d_Ad_B) +4)  \delta+2h(d_Ad_B\delta)$ and with the smallest non-vanishing eigenvalue $\lambda_{\min}=\min\{\lambda\in \Spec((T_p\otimes\id_2)(\varphi))| \lambda>0 \}$.
\end{theorem}

\begin{IEEEproof}[Proof sketch]
In the proof of \cite[Theorem~21.4.1]{Wilde13}, we construct a coding scheme such that the conditions \cite[Eq.~(21.61)-(21.64))]{Wilde13} are fulfilled. These conditions correspond to the conditions in \cite[Eq.~(16.70)-(16.73)]{Wilde13} with $\epsilon=e^{-\frac{n\delta^2}{(\log(\lambda_{\min}))^2}}$ (by Hoeffding's bound \cite{Hoeffding63}), $d=2^{n (H(A'B)_{(T_p\otimes \id_2)(\varphi)}  +2\delta)}$ (\cite[Property~15.1.2]{Wilde13}) and $D=2^{n (H(A')_{(T_p\otimes \id_2)(\varphi)} +H(B)_{(T_p\otimes \id_2)(\varphi)} -\omega(n,\delta)-2\delta}) $ (\cite[Property~15.1.3]{Wilde13}) with some function $\omega(n,\delta)$. If these conditions are fulfilled, the Packing Lemma \cite[Corollary~16.5.1]{Wilde13} (non-randomized version) applies and can be used to bound the error probability as
\begin{IEEEeqnarray*}{lCl}
   \epsilon_{ea}&\leq& 4 (\epsilon+2\sqrt{\epsilon}) +16 \frac{d}{D} 2^{n R'} \\&\leq& 12 \sqrt{\epsilon} + 16 \frac{d}{D} 2^{n R'}. 
\end{IEEEeqnarray*}
 Inserting the corresponding $\epsilon$, $d$ and $D$ from the coding scheme into the bound from the Packing Lemma leads to the following expression for the coding error:
\begin{IEEEeqnarray*}{lCl}\epsilon_{ea}&\leq &
12 e^{-\frac{n\delta^2}{2(\log(\lambda_{\min}))^2}}\\&&  + 16 \cdot 2^{-n(I(A':B)_{(T_p\otimes \id_2)(\varphi)} -4\delta-\omega(n,\delta) -R')}.\end{IEEEeqnarray*}

Furthermore, we have that $\omega(n,\delta)=d_Ad_B \delta \log(d_Ad_B) +h(d_Ad_B \delta  ) + d_Ad_B\frac{\log(n+1)}{n} $ as a consequence of \cite[Property~14.7.5]{Wilde13} and \cite[Eq.~14.119]{Wilde13} (note also that we are using $2\delta$ instead of just $\delta$ in \cite{Wilde13}). Grouping together terms, we define a function $\eta(\delta,d_A,d_B)=({d_Ad_B}\log(d_Ad_B) + 4)\delta +h(d_Ad_B \delta  )$.
\end{IEEEproof}

\section{Fault-tolerant entanglement distillation} \label{appendix-entdistproof}

Here, we give the proof of Theorem~\ref{thm-ft-ent-dist-normally} and discuss the version of the distillation protocol we use in the proof of our main result.

\begin{IEEEproof}[Proof of Theorem~\ref{thm-ft-ent-dist-normally}]
The distillation protocol $\Dist:\mathcal{M}_{2}^{\otimes 2k}\rightarrow \mathcal{M}_{2}^{\otimes 2\beta(q) k}$ is constructed as follows: For any quantum state $\sigma$, the first step of the distillation protocol transforms a quantum state $\psi_q=(1-q)\phi_+ +q\sigma$ to a Bell-diagonal quantum state $\phi_{q}=(1-q)\phi_+ + \frac{q}{3} (\phi_- +\psi_+ + \psi_-)$. Then, the protocol executes the steps from \cite{DW03} applied to this Bell-diagonal state. We denote the local operations performed by the sender by a quantum channel $\Dist_{\mathcal{E}}:\mathcal{M}_{2}^{\otimes k}\rightarrow \mathcal{M}_{2}^{\otimes  \beta(q) k} \otimes \mathbbm{C}^{2^{(1-\beta(q)) k }}$, and the local operations on the receiver's side are described by a quantum channel $\Dist_{\mathcal{D}}:\mathcal{M}_{2}^{\otimes k}\otimes \mathbbm{C}^{2^{(1-\beta(q)) k }}\rightarrow \mathcal{M}_{2}^{\otimes  \beta(q) k} $ where
\[\beta(q)=1-h_2(q)-q\log(3) .\]
They perform classical communication of $(1-\beta(q)) k$ bits between them, which is modelled by a connection via the classical identity channel $\id_{cl}^{\otimes (1-\beta(q))  k}$.
In total, as sketched in Figure \ref{fig-ft-ent-dist}, we have

\begin{IEEEeqnarray*}{lCl}
    \Dist &=& \Big(\id_2^{\otimes \beta(q) k} \otimes \Dist_{\mathcal{D}}\Big)\circ\Compactcdots \\&& \Compactcdots \circ  \Big( \id_{cl}^{\otimes (1-\beta(q)) k} \otimes \id_2^{\otimes k} \Big) \circ\Compactcdots \\&& \Compactcdots \circ \Big( {{\Dist_{\mathcal{E}}}}  \otimes \id_2^{\otimes k} \Big) . 
\end{IEEEeqnarray*}

These maps can be implemented in terms of specific gates using the Solovay-Kitaev theorem \cite{NC00}, with
\[\| \Dist_{\mathcal{E}}-\Gamma^{\Dist_{\mathcal{E}}} \|_{1\rightarrow1} \leq \frac{1}{k},\]
\[\| \Dist_{\mathcal{D}}-\Gamma^{\Dist_{\mathcal{D}}} \|_{1\rightarrow1}\leq \frac{1}{k}.\]

With these circuits as building blocks, we construct a circuit performing the distillation protocol from \cite{DW03} as
\begin{IEEEeqnarray*}{lCl}
\Gamma^{\Dist}&=  &\Big(\id_2^{\otimes \beta(q) k} \otimes \Gamma^{\Dist_{\mathcal{D}}} \Big) \circ\Compactcdots \\&& \Compactcdots \circ \Big( \id_{cl}^{\otimes (1-\beta(q)) k} \otimes \id_2^{\otimes k} \Big) \circ\Compactcdots \\&& \Compactcdots \circ \Big( \Gamma^{\Dist_{\mathcal{E}}}   \otimes \id_2^{\otimes k} \Big) .    
\end{IEEEeqnarray*}

This circuit has a fault-tolerant implementation using the concatenated 7-qubit Steane code $\mathcal{C}_l$ with concatenation level $l$ as introduced in Section \ref{sec-ft-circuits}, which we denote by $\Gamma^{\Dist}_{\mathcal{C}_l}$.

Because of \cite[Theorem~10]{DW03}, and using the Fuchs-van de Graaf inequalities \cite{FvdG97} to transform between fidelity and trace distance, we have
\begin{equation} \label{eq-dist0}
\|\Dist  (\psi_q^{\otimes k}) - \phi_+^{\otimes \beta(q) k}  \|_{\Tr} \leq  \sqrt{ \epsilon_{dist}(q) } \end{equation}
for quantum states of the form $\psi_q=(1-q)\phi_+ +q\sigma$ for some $\sigma$, where $\epsilon_{dist}(q)$ is the function from Theorem~\ref{thm-ent-dist-original}.

By triangle inequality, this implies that the circuit implementing the distillation protocol fulfills

\begin{equation}\label{eq-dist1} \|\Gamma^{\Dist} (\psi_q^{\otimes k}) - \phi_+^{\otimes \beta(q) k}  \|_{\Tr} \leq  \sqrt{ \epsilon_{dist}(q) } +\frac{2}{k}.
\end{equation}

Now, we make use of Lemma \ref{thm-effective-encoder} to justify that we can find fault-tolerant implementations of the distillation circuits in the concatenated 7-qubit Steane code with concatenation level $l$. For any $0\leq p \leq \frac{p_0}{2}$ and any $l\in\mathbbm{N}$, there exist quantum states $\sigma_S^{\mathcal{E}}$ and $\sigma_S^{\mathcal{D}}$ and a quantum channel $N_l:\mathcal{M}_2\rightarrow \mathcal{M}_2$ which only depends on $l$ and the interface circuit $\EncI_l$, such that
\begin{IEEEeqnarray*}{Cl}
            \IEEEeqnarraymulticol{2}{l}{ \|  ( (\DecI_l^*)^{\otimes \beta(4cp) k}\otimes \id_{cl}^{\otimes (1-\beta(4cp)) k} ) \circ \big[\Gamma_{\mathcal{C}_l}^{\Dist_{\mathcal{E}}} \circ (\EncI_l^{\otimes k}) \big]_{\mathcal{F}(p)} }\\&- (\Gamma^{\Dist_{\mathcal{E}}} \circ N_{enc,l,p}^{\otimes k})\otimes \sigma_S^{\mathcal{E}} \|_{1\rightarrow 1} \\\leq &2 p_0 (\frac{p}{p_0})^{2^l} |\Loc(\Gamma^{\Dist_{\mathcal{E}}})|     \end{IEEEeqnarray*}  
and
\begin{IEEEeqnarray*}{Cl}
            \IEEEeqnarraymulticol{2}{l}{ 
        \| (\DecI_l^*)^{\otimes \beta(4cp) k}) \circ \big[\Gamma_{\mathcal{C}_l}^{\Dist_{\mathcal{D}}} \circ (\EncI_l^{\otimes k} \otimes \id_{cl}^{\otimes (1-\beta(4cp))k} ) \big]_{\mathcal{F}(p)} } \\&- (\Gamma^{\Dist_{\mathcal{D}}}\circ (N_{enc,l,p}^{\otimes k}\otimes \id_{cl}^{\otimes (1-\beta(4cp))k}) ) \otimes \sigma_S^{\mathcal{D}} \|_{1\rightarrow 1}  \\ \leq &2 p_0 (\frac{p}{p_0})^{2^l} |\Loc(\Gamma^{\Dist_{\mathcal{D}}})|
    \end{IEEEeqnarray*}   
with
\[N_{enc,l,p} = (1-2cp) \id_2 +2cpN_l\]
where $c=p_0 \max \{|\Loc(\EncI_1)|,|\Loc(\DecI_1\circ EC)|\}$.

In combination, we thereby obtain
\begin{IEEEeqnarray}{Cl}
  \IEEEeqnarraymulticol{2}{l}{\|(\DecI_l^*)^{\otimes 2\beta(q)k}\circ \big[ \Gamma^{\Dist}_{\mathcal{C}_l} \circ \EncI_l^{\otimes 2k} \big]_{\mathcal{F}(p)} (\phi_+^{\otimes k} )}\IEEEnonumber\\&-\Gamma^{\Dist}
    \circ (N_{enc,l,p}^{\otimes 2k}  (\phi_+^{\otimes k})) \otimes \sigma_S^{\mathcal{E}}\otimes \sigma_S^{\mathcal{D}}
    )\|_{\Tr} \label{eq-dist2}\\ \leq &  2 p_0 (\frac{p}{p_0})^{2^l} (|\Loc(\Gamma^{\Dist_{\mathcal{E}}})|+|\Loc(\Gamma^{\Dist_{\mathcal{D}}})|) \IEEEnonumber\\\leq &2 p_0 (\frac{p}{p_0})^{2^l} |\Loc(\Gamma^{\Dist})| \IEEEnonumber
   \end{IEEEeqnarray}

with a syndrome state $\sigma_S=\sigma_S^{\mathcal{E}}\otimes \sigma_S^{\mathcal{D}}$ that is a product state in the cut between the sender and the receiver.

Effectively, the physical input states, which are maximally entangled states here, are transformed as follows:
 
\begin{IEEEeqnarray*}{lCl}
             N_{enc,l,p}^{\otimes2}(\phi_+) &=& (1-2cp)^2 \phi_+ +(1-(1-2cp)^2) \sigma_l \\&\leq &(1-4cp) \phi_+ + 4cp \sigma_l\\& =: &\psi_{4cp} \end{IEEEeqnarray*}
with some quantum state $\sigma_l$, where we use Bernoulli's inequality. 

In total, combining Eq.~\eqref{eq-dist1} and Eq.~\eqref{eq-dist2}, and using the triangle inequality for trace distance (where we use that $\sigma_S^{\mathcal{E}}\otimes \sigma_S^{\mathcal{D}}$ is a normalized quantum state), we find 
\begin{IEEEeqnarray*}{Cl}
            \IEEEeqnarraymulticol{2}{l}{ \|(\DecI_l^*)^{\otimes 2\beta(4cp)k}\circ \big[ \Gamma^{\Dist}_{\mathcal{C}_l} \circ \EncI_l^{\otimes 2k} \big]_{\mathcal{F}(p)} (\phi_+^{\otimes k})} \\& -\phi_+^{\otimes \beta(4cp)k}\otimes \sigma_S^{\mathcal{E}}\otimes \sigma_S^{\mathcal{D}}\|_{\Tr} \\
   \leq &
    \|(\DecI_l^*)^{\otimes 2\beta(4cp)k}\circ \big[ \Gamma^{\Dist}_{\mathcal{C}_l} \circ \EncI_l^{\otimes 2k} \big]_{\mathcal{F}(p)} (\phi_+^{\otimes k}) \\&  - (\Gamma^{\Dist}
    \circ (N_{enc,l,p}^{\otimes 2k} (\phi_+^{\otimes k})) \otimes \sigma_S^{\mathcal{E}}\otimes \sigma_S^{\mathcal{D}}
  \|_{\Tr} \\&  + \|\Gamma^{\Dist} 
    \circ (N_{enc,l,p}^{\otimes 2k} (\phi_+^{\otimes k})) - \phi_+^{\otimes \beta(4cp)k}  \|_{\Tr} 
    \\\leq& 2 p_0 (\frac{p}{p_0})^{2^l} |\Loc(\Gamma^{\Dist})|  + \|\Gamma^{\Dist}
    (\psi_{4cp}^{\otimes k}) - (\phi_+)^{\otimes \beta(4cp)k} \|_{\Tr}
    \\
   \leq &  2 p_0 (\frac{p}{p_0})^{2^l} |\Loc(\Gamma^{\Dist})|  +  \sqrt{ \epsilon_{dist}(4cp) } + \frac{2}{k}.
 \end{IEEEeqnarray*}

\end{IEEEproof}


 The proof of our main theorem \ref{thm-final-coding-thm} actually employs a modification of Theorem~\ref{thm-ft-ent-dist-normally}, where the entanglement distillation is performed with classical communication via a subset of the copies of the channel $T$, and an additional step of entanglement dilution to distill the bipartite pure entangled state $\varphi$ using the protocol from \cite{LP99}. Thereby, the fault-tolerant implementation of this distillation circuit distills a pure state $\varphi$ in the code space, where it is later used for communication based on the protocol from \cite{Wilde13}.

\begin{theorem}[{Entanglement dilution, \cite{LP99}, \cite[Eq.~(19.113)]{Wilde13}}] \label{thm-ent-dil-original}
For sufficiently large $k\in\mathbbm{N}$, there exists a $\zeta\geq0$ and a quantum channel $\Dil:\mathcal{M}_2^{\otimes H(A)_{\varphi} 2k}\rightarrow \mathcal{M}_2^{\otimes 2k}$ consisting of local operations and $\propto \zeta$ bits of one-way classical communication, such that $H(A)_{\varphi}k$ copies of the state $\phi_{+}\in \mathcal{M}_{2}\otimes \mathcal{M}_{2}$ are mapped to $k$ copies of the state $\varphi\in \mathcal{M}_{2}\otimes \mathcal{M}_{2}=:\mathcal{M}_{d_A}\otimes \mathcal{M}_{d_{A'}}$
with the following fidelity:
\[ F\big(\varphi^{\otimes k}, \Dil(\phi_+^{\otimes  H(A)_{\varphi} k}) \big) \geq 1- \epsilon_{dil}(q)\]
with 
\[\epsilon_{dil}(q) \leq 2 e^{-\frac{k \zeta^2}{2 \log(\lambda_{min}(\varphi) )^2} } + \sqrt{ 2 \cdot 2^{-k \zeta }}.\]
\end{theorem}

With this, we have the following adaptation of Theorem~\ref{thm-ft-ent-dist-normally} with an additional dilution step performed together with the entanglement distillation:

\begin{corollary}[Entanglement distillation with an extra dilution step]\label{thm-ent-dist-and-dil}
    For each $l\in \mathbbm{N}$, let $\mathcal{C}_l$ denote the $l$-th level of the concatenated 7-qubit Steane code with threshold $p_0$.
For any $0\leq p \leq \frac{p_0}{2}$ and for all $k\in \mathbbm{N}$ large enough, there exists a circuit 
$\Gamma^{\Dist}:\mathcal{M}_2^{\otimes 2k}\rightarrow \mathcal{M}_2^{\otimes 2\beta(4cp)k}$
 using $(1-\beta(4cp))k$ bits of classical communication, and two quantum states $\sigma^{\mathcal{E}}_S \in \mathcal{M}_{d_{S_E}}$ and $\sigma^{\mathcal{D}}_S \in \mathcal{M}_{d_{S_D}}$ such that
         \begin{IEEEeqnarray*}{ll}
            \IEEEeqnarraymulticol{2}{l}{\|(\DecI_l^*)^{\otimes 2\beta(4cp)k}\circ \big[ \Gamma^{\Dist}_{\mathcal{C}_l}  \circ \EncI_l^{\otimes 2k} \big]_{\mathcal{F}(p)} (\phi_+)^{\otimes k}} \\ & -(\varphi)^{\otimes \beta(4cp)k}  \otimes \sigma^{\mathcal{E}}_S \otimes \sigma^{\mathcal{D}}_S\|_{\Tr}   \\ \leq&  p_0 \left(\frac{p}{p_0}\right)^{2^l} |\Loc(\Gamma^{\Dist})|  +  \sqrt{ \epsilon'_{dist}(4cp) } + \frac{2}{k} 
        \end{IEEEeqnarray*}
with the constant $c$ from Theorem~\ref{thm-correct-interfaces}, and $\beta(q)=\frac{1-h(q)-q\log(3)}{H(A)_{\varphi}}$, and 
\[\epsilon'_{dist}(4cp)\leq  \epsilon_{dist}(4cp) + \epsilon_{dil}(4cp) \]
where $\epsilon_{dist}(q)$ is the function from Theorem~\ref{thm-ent-dist-original} and $\epsilon_{dil}(q)$ is the function from Theorem~\ref{thm-ent-dil-original}.
\end{corollary}


\IEEEtriggeratref{18}
\bibliographystyle{IEEEtran} \bibliography{references} 
\newpage
\begin{IEEEbiographynophoto}{Paula Belzig}
 received the Ph.D. degree from the University of Copenhagen in 2023. She is currently a post-doctoral researcher at the Institute for Quantum Computing at the University of Waterloo. Her research focuses on the theory of quantum communication with and without entanglement, and with and without noise assumptions on the communication setup.
\end{IEEEbiographynophoto}
\begin{IEEEbiographynophoto}{Matthias Christandl}
 received the Ph.D. degree from the University of Cambridge. He is currently a Professor with the Department of Mathematical Sciences, University of Copenhagen, Denmark, where he leads the Quantum for Life Center. He has previously held faculty positions in Munich and at ETH Zürich. His research interests include quantum information theory and quantum computation. He is a member of the Royal Danish Academy of Sciences and Letters.
\end{IEEEbiographynophoto}
\begin{IEEEbiographynophoto}{Alexander Müller-Hermes}
 received the Ph.D. degree from the Technical University of Munich in 2015. After being in a post-doctoral position at the Centre of Mathematics in Quantum Theory (QMATH), University of Copenhagen, he obtained a Marie Skłodowska-Curie Fellowship with the Institute Camille Jordan, Université Claude Bernard Lyon 1. Since 2021, he has been an Associate Professor with the Department of Mathematics, University of Oslo. His research interests include the mathematical aspects of quantum information theory, quantum Shannon theory, cake cutting, and entanglement theory.
\end{IEEEbiographynophoto}

\end{document}